\documentclass{article}
\usepackage[margin=1in]{geometry}

\usepackage[colorlinks,citecolor=blue,breaklinks]{hyperref}

\usepackage{natbib}  %
\usepackage{caption} %

\usepackage{amssymb}
\usepackage[ruled]{algorithm2e} %

\usepackage{amsfonts}       %
\usepackage{nicefrac}       %
\usepackage{microtype}      %
\usepackage{float}
\usepackage{multicol}
\usepackage{multirow}
\usepackage{siunitx}
\usepackage{tikz}
\usepackage{pgf}
\usepackage{mleftright}
\usepackage{enumitem}
\usepackage{etoolbox}
\usepackage{contour}
\usepackage{mdframed}

\usepackage{natbib}
\usepackage{amsmath}
\usepackage{amsthm}

\usepackage{bm}

\usepackage{thm-restate}
\usepackage{lipsum}

\usepackage{mathtools}

\usepackage[capitalise,noabbrev]{cleveref}
\crefname{property}{Property}{Properties}
\crefname{example}{Example}{Examples}

\usepackage{csquotes}

\usepackage{nicerbb}

\newtheorem{theorem}{Theorem}[section]
\newtheorem{lemma}[theorem]{Lemma}
\newtheorem{corollary}[theorem]{Corollary}

\newtheorem{assumption}[theorem]{Assumption}

\newtheorem{definition}[theorem]{Definition}

\newtheorem{remark}[theorem]{Remark}

\newcommand*{\bbQ}{{\mathbb{Q}}}

\newcommand*{\bbR}{{\mathbb{R}}}

\newcommand*{\cJ}{{\mathcal{J}}}
\newcommand*{\cK}{{\mathcal{K}}}

\newcommand*{\cI}{{\mathcal{I}}}

\let\poly\relax

\let\co\relax

\DeclareMathOperator{\poly}{poly}

\DeclareMathOperator{\co}{co}

\renewcommand{\vec}[1]{\bm{#1}}
\newcommand*{\mat}[1]{\mathbf{#1}}

\newcommand*{\range}[1]{[\kern-1mm[#1]\kern-.9mm]}

\makeatletter
\tikzset{
  fitting node/.style={
      inner sep=0pt,
      fill=none,
      draw=none,
      reset transform,
      fit={(\pgf@pathminx,\pgf@pathminy) (\pgf@pathmaxx,\pgf@pathmaxy)}
    },
  reset transform/.code={\pgftransformreset}
}
\tikzset{cross/.style={path picture={
          \draw[black]
          (path picture bounding box.south east) -- (path picture bounding box.north west) (path picture bounding box.south west) -- (path picture bounding box.north east);
        }}}
\tikzstyle{ox}=[semithick,draw=black,circle,cross,inner sep=1.2mm]
\makeatother

\newcommand{\declarecolor}[2]{\definecolor{#1}{RGB}{#2}\expandafter\newcommand\csname #1\endcsname[1]{\textcolor{#1}{##1}}}
\declarecolor{White}{255, 255, 255}
\declarecolor{Black}{0, 0, 0}
\declarecolor{Maroon}{128, 0, 0}
\declarecolor{Coral}{255, 127, 80}
\declarecolor{Red}{182, 21, 21}
\declarecolor{LimeGreen}{50, 205, 50}
\declarecolor{DarkGreen}{0, 100, 0}
\declarecolor{Navy}{0, 0, 128}

\NewDocumentCommand{\numberthis}{om}{%
  \IfNoValueTF{#1}{%
    \refstepcounter{equation}\tag{\theequation}%
  }{%
    \tag{#1}%
  }%
  \label{#2}%
}

\newcommand{\timehat}[1]{^{(#1)}}

\newtoks\mymathaccents
\mymathaccents={%
\let\^\timehat
}

\everymath=\expandafter\expandafter\expandafter{%
  \expandafter\the\expandafter\everymath\the\mymathaccents}

\everydisplay=\expandafter\expandafter\expandafter{%
  \expandafter\the\expandafter\everydisplay\the\mymathaccents}

\def\[#1\]{%
\begin{flalign*}#1\end{flalign*}%
}

\usepackage{dsfont}

\newcommand{\vone}{\vec{1}}

\usetikzlibrary{fit,shapes.misc,arrows.meta,calc,positioning}
\makeatletter
\tikzset{
  fitting node/.style={
      inner sep=0pt,
      fill=none,
      draw=none,
      reset transform,
      fit={(\pgf@pathminx,\pgf@pathminy) (\pgf@pathmaxx,\pgf@pathmaxy)}
    },
  reset transform/.code={\pgftransformreset}
}
\tikzset{cross/.style={path picture={
          \draw[black]
          (path picture bounding box.south east) -- (path picture bounding box.north west) (path picture bounding box.south west) -- (path picture bounding box.north east);
        }}}
\tikzstyle{ox}=[semithick,draw=black,circle,cross,inner sep=1.2mm]
\pgfdeclarelayer{background}
\pgfdeclarelayer{middle}
\pgfsetlayers{background,middle,main}

\makeatother

\crefname{property}{Property}{Properties}

\tikzset{cross/.style={path picture={
          \draw[black]
          (path picture bounding box.south east) -- (path picture bounding box.north west) (path picture bounding box.south west) -- (path picture bounding box.north east);
        }}}
\tikzstyle{chanode}   = [fill=white,draw=black,circle,cross,inner sep=.8mm]
\tikzstyle{pl1node}   = [fill=black,draw=black,circle,inner sep=.55mm]
\tikzstyle{pl2node}   = [fill=white,draw=black,circle,inner sep=.55mm]
\tikzstyle{termina}   = [fill=white,draw=black,inner sep=.6mm]
\tikzstyle{decpt}     = [fill=black,draw=black,inner sep=.8mm]
\tikzstyle{obspt}     = [fill=white,draw=black,cross,inner sep=0.95mm]
\tikzstyle{highlight} = [line width=1.99]
\tikzstyle{infoset} = [black!50!white]

\newcommand{\cQ}{\ensuremath{\mathcal{Q}}}

\newcommand{\cH}{\ensuremath{\mathcal{H}}}

\newcommand{\cX}{\mathcal{X}}

\let\olabel\label
\NewDocumentCommand \constraint {o} {%
  {\refstepcounter{equation}\mathrm{(\theequation)}\IfValueT{#1}{\olabel{#1}}}
}

\DeclareMathOperator*{\argmax}{arg\,max}

\newcommand*\circled[1]{\refstepcounter{equation}\mathrm{(\theequation)}}

\NewDocumentCommand  \Pure           {o} {{\Pi\IfNoValueF{#1}{_{#1}}}}
\NewDocumentCommand  \PureSub       {om} {\Pi_{\IfNoValueF{#1}{#1,\,}\succeq\,#2}}
\NewDocumentCommand  \Seqf           {o} {{\cQ\IfNoValueF{#1}{_{#1}}}}
\NewDocumentCommand  \SeqfSub       {om} {\cQ_{\IfNoValueF{#1}{#1,\,}\succeq\,#2}}
\NewDocumentCommand  \Seqs          {so} {{\Sigma\IfBooleanT{#1}{^*}\IfNoValueF{#2}{_{#2}}}}
\NewDocumentCommand  \SeqsSub       {om} {\Sigma_{\IfNoValueF{#1}{#1,\,}\succeq\,#2}}
\NewDocumentCommand  \Infos          {o} {{\cI\IfNoValueF{#1}{_{#1}}}}
\NewDocumentCommand  \DecNodes       {o} {{\cJ\IfNoValueF{#1}{_{#1}}}}
\NewDocumentCommand  \ObsNodes       {o} {{\cK\IfNoValueF{#1}{_{#1}}}}

\NewDocumentCommand  \Hist           {o} {\cH\IfNoValueF{#1}{_{#1}}}

\NewDocumentCommand  \emptyseq        {} {{\varnothing}}

\DeclareMathOperator*{\Ex}{\mathbb{E}}

\usepackage{needspace}%

\AtBeginEnvironment{theorem}{\Needspace{5\baselineskip}}%

\title{%
Polynomial-Time Computation of Exact $\Phi$-Equilibria in Polyhedral Games
}

\author{
    Gabriele Farina\\
    MIT\\
    \texttt{gfarina@mit.edu}\\
    \and
    Charilaos Pipis\\
    MIT\\
    \texttt{chpipis@mit.edu}\\
}

\begin{document}

\maketitle

\begin{abstract}
    It is a well-known fact that correlated equilibria can be computed in polynomial time in a large class of concisely represented games using the celebrated Ellipsoid Against Hope algorithm \citep{Papadimitriou2008:Computing, Jiang2015:Polynomial}. However, the landscape of efficiently computable equilibria in sequential (extensive-form) games remains unknown. 
    The Ellipsoid Against Hope does not apply directly to these games, because they do not have the required ``polynomial type'' property. Despite this barrier, \citet{Huang2008:Computing} altered the algorithm to compute exact extensive-form correlated equilibria.

    In this paper, we generalize the Ellipsoid Against Hope and develop a simple algorithmic framework for efficiently computing saddle-points in bilinear zero-sum games, even when one of the dimensions is exponentially large. Moreover, the framework only requires a ``good-enough-response'' oracle, which is a weakened notion of a best-response oracle.

    Using this machinery, we develop a general algorithmic framework for computing exact linear $\Phi$-equilibria in any polyhedral game (under mild assumptions),
    including correlated equilibria in normal-form games, and extensive-form correlated equilibria in extensive-form games. This enables us to give the first polynomial-time algorithm for computing exact linear-deviation correlated equilibria in extensive-form games, thus resolving an open question by \citet{Farina2023:Polynomial}. Furthermore, even for the cases for which a polynomial time algorithm for exact equilibria was already known, our framework provides a conceptually simpler solution.

\end{abstract}

\vspace{1cm}
\setcounter{tocdepth}{2} %
\newpage
\tableofcontents
\newpage

\section{Introduction}

The correlated equilibrium (CE), introduced by \citet{Aumann1974:Subjectivity}, is one of the most seminal solution concepts in multi-player games. Contrary to the Nash equilibrium, in a correlated equilibrium the players' strategies are correlated by a fictitious \emph{mediator} that can recommend (but not enforce) behavior. It is then up to this mediator to ensure that the distribution of recommendations does not incentivize any player to \emph{deviate} from their recommended strategy. It is known that this type of equilibrium naturally emerges from the repeated interaction of learning agents \citep{Hart00:Simple}. In practice, this means that one can compute $\epsilon$-approximate CEs in normal-form games by implementing suitable decentralized no-regret dynamics, of which several efficient implementations are known (see, \emph{e.g.}, \citet{Blum2007:From} and \citet{Anagnostides21:NearOptimal}). However, this approach requires $\Omega(\textrm{poly}(1/\epsilon))$ iterations to compute an $\epsilon$-approximate equilibrium, making it a non-viable choice for high-precision equilibrium computation.
In a celebrated result, \citet{Papadimitriou2008:Computing} (with later refinements by  \citet{Jiang2015:Polynomial}) devised an algorithm, called \emph{Ellipsoid Against Hope}, that can compute an exact CE in a concisely represented normal-form game in polynomial time in the representation of the game. Their algorithm is an algorithmic version of the clever reduction of \citet{Hart1989:Existence} that casts the computation of CEs as a two-player zero-sum game. 

The positive results for normal-form games, however, do not transfer directly to the significantly more involved setting of \emph{extensive-form games}. Extensive-form games are games played on a game tree and can model sequential and simultaneous moves, as well as imperfect information. Despite a significant stream of \emph{positive} results related to learning and equilibrium computation in extensive-form games, the complexity of computing CEs in extensive-form games remains to this day a major open question \citep{Farina2023:Polynomial,vonStengel2008}. 
Due to its conjectured intractability, researchers have resorted to considering the computation of weaker and generalized notions of correlated equilibrium. A key reference point in this space is given by the framework of \citet{Gordon08:No}, who define a generalized notion of CE called $\Phi$-equilibria. In a $\Phi$-equilibrium, every player $p$ is endowed with a set $\Phi_p$ of \emph{behavior transformation functions}. The goal of the fictitious mediator is then simply to recommend strategies such that no player could unilaterally benefit from deviating by using any of the functions $\phi \in \Phi_p$. In this language, a CE corresponds to the $\Phi$-equilibrium in which each $\Phi_p$ is the set of all possible functions from the strategy set of the player to itself. However, by considering appropriate subsets of behavior transformations, weaker supersets of CEs can be efficiently computed and learned through uncoupled learning dynamics. Notable examples of such equilibria in extensive-form games include the extensive-form correlated equilibrium (EFCE) \citep{vonStengel2008}, the extensive-form coarse-correlated equilibrium (EFCCE) \citep{Farina20:Coarse}, the normal-form coarse-correlated equilibrium (NFCCE) \citep{Moulin1978:Strategically}, the recently-introduced linear-deviation correlated equilibrium (LCE) \citep{Farina2023:Polynomial}, and others \citep{Morrill2021:Efficient}.

\citet{Huang2008:Computing} proposed a specialization of the Ellipsoid Against Hope algorithm to compute exact EFCE in extensive-form games. Later, \citet{Farina22:Simple} showed efficient no-regret dynamics that converge to the EFCE. More recently, there has been increased interest in understanding what is the $\Phi$-equilibrium that is the closest to CE while still enabling efficient computation and learning. \citet{Farina2023:Polynomial} introduced the linear-deviation correlated equilibrium (LCE) that arises from the set $\Phi_\text{LIN}$ of all linear-swap deviations in sequence-form strategies and devise efficient no-linear-swap regret dynamics to approximate it. The LCE captures all notable notions of equilibrium that were previously known to be efficiently computable (including EFCE, EFCCE, and NFCCE). However, \citet{Farina2023:Polynomial} left open the key question as to whether LCEs themselves can also be computed exactly in polynomial time, akin to the generalization of the Ellipsoid Against Hope algorithm by \citet{Huang2008:Computing}
, as opposed to just learned via uncoupled learning dynamics. The approach by \citet{Huang2008:Computing} relies heavily on the combinatorial structure of the deviation functions that define EFCE, resulting in a rather involved algorithm. This is in stark contrast to the simple framework for constructing $\Phi$-regret minimizers championed by \citet{Gordon08:No}. This begs the natural question:
\begin{displayquote}\itshape\centering
    Can we always construct an efficient algorithm for exactly computing $\Phi$-equilibria, when there exists an efficient no-$\Phi$-regret minimizer?
\end{displayquote}
In other words, can we create a simple and general framework in the spirit of \citet{Gordon08:No} that can enable us to construct algorithms for the exact computation of $\Phi$-equilibria in polyhedral games for any $\Phi \subseteq \Phi_\text{LIN}$? We answer this question in the affirmative.

\subsection{Contributions}

In this paper, we propose a framework for computing exact $\Phi$-equilibria in general polyhedral games. Our framework recovers all positive results established by \citet{Papadimitriou2008:Computing}, and crucially applies to polyhedral games such as extensive-form games. %
Using our framework, we develop the first polynomial-time algorithm for computing exact linear-deviation correlated equilibria in extensive-form games, thus resolving an open question by \citet{Farina2023:Polynomial}. Furthermore, even for the cases for which a polynomial time algorithm for exact equilibria was already known (CEs in normal-form games \citep{Papadimitriou2008:Computing, Jiang2015:Polynomial} and EFCEs in extensive-form games \citep{Huang2008:Computing}), our framework provides a conceptually simpler solution.

We show that to compute an exact $\Phi$-equilibrium in a polyhedral game, the following three conditions are sufficient:
\begin{enumerate}
    \item The game satisfies the ``polynomial utility gradient property'' (\cref{asmpt:utility_gradient}) which states that given a product distribution over the joint strategy space of all $n$ players, we can efficiently compute the expectation of the gradient of any player's utility. This is a natural generalization of the ``polynomial expectation property'' from \citet{Papadimitriou2008:Computing}, and it is a rather low bar to clear (in fact, it is implicitly assumed in every no-regret learning algorithm).

    \item The set of transformations $\Phi$ contains valid linear transformations that map the strategy set to itself. This is a technical requirement so that the expectation operator and the application of the deviation function can commute.\footnote{Note that going beyond linear transformations can introduce several complications. Most notably, in a recent paper, \citet{Zhang2024:Efficient} observe that computing exact fixed-points of non-linear transformations might be PPAD-hard and they instead introduce a new way to perform regret minimization using ``approximate expected fixed-points''.} This condition is satisfied by all notions of $\Phi$-equilibrium mentioned above, including EFCE and LCE in extensive-form games, and CE in normal-form games.

    \item The set $\Phi$ of transformations is a polytope that admits a polynomial-time separation oracle.
\end{enumerate}
The separation oracle requirement in the third condition is known to be equivalent to efficient linear optimization \citep{Grotschel1993:Geometric}.
In essence, this means that giving a polynomially-sized characterization of a set $\Phi$ of linear transformations for a polyhedral set is a sufficient condition to provide both efficient no-regret learning dynamics and an efficient algorithm for computing exact $\Phi$-equilibria. 
This is exactly what we achieve by applying our result to the set of linear-swap deviations that was recently characterized \citep{Farina2023:Polynomial, Zhang2024:Mediator} as a polytope of polynomially many constraints and was used to prove efficient no-linear-swap regret dynamics.
In light of all these considerations, our framework can be thought of as the counterpart of the $\Phi$-regret minimization framework by \citet{Gordon08:No}, but for computation of \textit{exact} equilibria rather than regret minimization.

At the heart of our construction, our main technical tool is a generalization of the methodology of the Ellipsoid Against Hope \citep{Papadimitriou2008:Computing, Jiang2015:Polynomial} to general polyhedral bilinear games. In more detail, we give a new constructive proof of the minimax theorem for players with polyhedral strategy sets, by using only a weakened type of a best-response oracle that we coin Good-Enough-Response (GER) oracle. An interesting property of the GER oracle is that it can be computationally tractable even when the respective best-response oracle is intractable, as we show in \cref{sec:exact_phi}. This algorithmic idea is likely of independent interest and is especially useful when the strategy space of one of the players is very large but there exists an efficient GER oracle that outputs sparse solutions (\textit{e.g.}, vertices of a high-dimensional polytope). This is exactly the type of problem we face when we need to compute exact $\Phi$-equilibria in polyhedral games and we then proceed to apply this machinery to the above question.
Interestingly, in order to show the existence of structured good-enough responses in the context of $\Phi$-equilibria, we use an argument based on the existence of an efficient fixed-point oracle for each deviation $\phi \in \Phi$. Such an ingredient was fundamental (albeit used differently; see \citet{Hazan2007:Computational} for a discussion of the role played by fixed-point oracles in the construction of no-$\Phi$-regret algorithms) also in \citet{Gordon08:No}. In our case, it is one of the technical insights that enable us to sidestep much of the intricacy encountered by \citet{Huang2008:Computing}.

\subsection{Related work}

Algorithms for computing equilibria can be classified broadly into three categories: 

\begin{itemize}[nosep,left=0mm]
    \item \textit{Polynomial-time algorithms} compute an \emph{exact} equilibrium in time polynomial in the input game size. Note that exact equilibria only make sense when the game has rational utilities, otherwise we settle for $\epsilon$-approximate equilibria in time polynomial in $\log(1/\epsilon)$ and the size of the game.

    \item \textit{Fully polynomial-time approximation schemes (FPTASs)}
    compute $\epsilon$-approximate equilibria in time polynomial in $1/\epsilon$ and the size of the input game.

    \item \textit{Polynomial-time approximation schemes (PTASs)} compute $\epsilon$-approximate equilibria in time that is polynomial in the size of the input game, for every fixed $\epsilon > 0$---however, they might in general have an exponential dependence on $1/\epsilon$.
\end{itemize}
Nash equilibria are known to be PPAD-complete, thus ruling out any polynomial-time algorithm for them \citep{Daskalakis2009:The}. Additionally, approximating a Nash equilibrium, even for a constant $\epsilon$ is known to also be PPAD-complete for $n$-player games \citep{Rubinstein2015:Inapproximability} and to require quasi-polynomial time for 2-player games, assuming ETH for PPAD \citep{Rubinstein2016:Settling}, thus ruling out any PTAS or FPTAS algorithm.

The complexity landscape is significantly more favorable for the case of correlated equilibrium (CE). Specifically, \citet{Hart00:Simple,Blum2007:From} gave efficient no-regret dynamics (minimizing the so-called \emph{internal regret}) that, if used by all players in a game, can be used to compute an $\epsilon$-approximate CE in normal-form games in time polynomial in the size of the game and $1/\epsilon$. This constitutes the first FPTAS for the computation of CEs. Finally, \citet{Papadimitriou2008:Computing} gave a centralized algorithm that exactly computes a CE in a concisely represented normal-form game in polynomial time in the size of the game. This was the first polynomial-time algorithm for CEs.

The complexity of CE is not settled for extensive-form games (EFGs), and its determination remains a major unresolved question in the field \citep{vonStengel2008,Farina2023:Polynomial}. An advance in this direction was provided very recently by the breakthrough results from two concurrent works of \citet{Dagan2023:From} and \citet{Peng2023:Fast}, which provide a PTAS for CE in extensive-form games, though it remains unclear whether a polynomial-time algorithm or even an FPTAS exist. 

Due to the conjectured intractability of computing a normal-form CE in EFGs, researchers have come up with other notions of equilibrium \citep{vonStengel2008, Morrill2021:Efficient, Farina2023:Polynomial, Zhang2024:Mediator} that lie on a spectrum of $\Phi$-equilibria. $\Phi$-equilibria provide a generalization of CE that ranges from least hindsight-rational (coarse correlated equilibria), to maximum hindsight rational (CE), depending on the size of the set of behavior transformation $\Phi$ considered by the players. One of the most notable and natural notions of sequential rationality is that of the extensive-form correlated equilibrium (EFCE) \citep{vonStengel2008}. The EFCE was shown to be efficiently computable exactly \citep{Huang2008:Computing} using a method similar to the Ellipsoid Against Hope of \citet{Papadimitriou2008:Computing}. Additionally, it was shown that there exist efficient regret dynamics (minimizing the so called trigger regret) that can be used to compute an $\epsilon$-approximate equilibrium by \citet{Farina22:Simple}. This also gives an FPTAS for EFCE.

Currently, the highest notion of rationality that admits an FPTAS is the recently defined linear-deviation correlated equilibrium (LCE) \citep{Farina2023:Polynomial}, which subsumes previous equilibrium notions such as EFCE. \citet{Farina2023:Polynomial} proved that there exist efficient regret dynamics (minimizing linear-swap regret) that can converge to an LCE; the polynomial complexity was then later improved by \citet*{Zhang2024:Mediator}. They however left open the question as to whether there exists a polynomial time algorithm that can compute an exact LCE. We resolve this open question in this paper, showing that the exact computability of equilibria in EFGs extends up to the linear-deviation correlated equilibrium.

\paragraph{Relationship with work on combinatorially-structured games.}
Our algorithmic framework in \cref{sec:eah_exact} can be used to compute exact min-max equilibria in bilinear games when one of the players has an exponentially large action space and we can only use a good-enough-response oracle---a weaker notion than the best-response oracle. 

In light of applications related to Machine Learning and Deep Learning, there has recently been increased renewed interest in games having exponential (or even infinite) action spaces. For example, \citet{Assos2023:Online} propose regret dynamics that can converge to approximate coarse correlated equilibria in infinite (nonparametric) games when a suitable notion of dimension of the game (Littlestone and fat-threshold) is bounded. \citet{Dagan2023:From} generalize this result even further proving that there also exist regret-dynamics in these cases that can converge in approximate correlated equilibria. Interestingly, recent breakthroughs in Large Language Models have inspired work on ``language-based'' games that typically have an enormous number of strategies \citep{Meta2022:Human, Gemp2024:States}.

However, the interest in games involving large strategy spaces is by no means a recent phenomenon.
For instance, the community of security games has traditionally been interested in the problem of computing Stackelberg equilibria for games where one player (the leader) can possibly have exponentially many strategies \citep{Kiekintveld2009:Computing, Xu2016:The}. Another notable example is that of the Colonel Blotto game, which involves exponentially large strategy sets in both players \citep{Ahmadinejad2019:From}. Finally, work on \emph{learning} in combinatorially-structured games has found applications to online optimization on combinatorial domains such as EFG strategy spaces and flow polytopes~\citep{Farina22:Kernelized,Koolen10:Hedging,Takimoto03:Path}.

Even though our framework for computing min-max equilibria might have some resemblance to some of the methods in these papers, to the best of our knowledge, in all of the past cases the algorithms cleverly exploit the special combinatorial structure inherent in security games and usually involve reducing the dimensionality of the strategy space as a first step.
Our framework on the other hand, might allow for applications where the large decision set of a player is not amenable to some kind of smaller-dimensional representation.

Finally, we remark that an interesting recurring theme in games with large strategy spaces is that they often assume some kind of best-response oracle access.

\paragraph{Relationship with work based on best-response oracle access to games}

Using best-response oracles is a ubiquitous technique for learning in games or equilibrium computation, starting from the foundational method of fictitious play \citep{Brown1951:Iterative} where players apply a best-response oracle at every round to respond to the empirical frequency of play of their opponent. Best-response oracles (or variants thereof) have additionally been used in security games \citep{Ahmadinejad2019:From, Xu2014:Solving}, in bilinear games \citep{Gidel2017:Frank}, in efficient learning on polytopes \citep{Chakrabarti2023:Efficient}, in the computation of well-supported equilibria in bilinear games \citep{Goldberg2021:Learning}, in the PSRO for Reinforcement Learning \citep{Lanctot2017:Unified}, in infinite games \citep{Assos2023:Online, Dagan2023:From}.
We remark however that in this paper we do \textit{not} use a best-response oracle in our algorithms. Rather, we use a weaker notion that we coin ``Good-Enough-Response'' (GER) oracle. In certain cases (such as the computation of $\Phi$-equilibria in \cref{sec:exact_phi}) it is critical to relax the requirement for a best-response oracle because no such oracle can be constructed unless P = NP (\cref{thm:efg_br_hard}).

\section{Preliminaries}

In this section, we introduce some basic concepts and definitions that will be used in developing our framework.

\subsection{Polyhedra, polytopes, and convex sets}

\begin{definition}[Rational polyhedron]
    A rational polyhedron $\mathcal{P} = \{\vec{x} \in \bbR^n \mid \mat{A} \vec{x} \leq \vec{b} \}$ is the solution set of a system of linear inequalities with rational coefficients. We say that $\mathcal{P}$ has \textbf{facet-complexity} $\varphi$ if there exists a system of linear inequalities, where each inequality has encoding length\footnote{The encoding length of an object is simply the amount of bits required to encode it. In this case, the encoding length of an inequality is the total amount of bits required to represent all of its coefficients.} at most $\varphi$, and whose solution set is $\mathcal{P}$. A rational polyhedron that is \emph{bounded} is called a \emph{rational polytope}.
\end{definition}

One important property of rational polytopes that we will use repeatedly throughout the paper is that they can equivalently be written as the convex hull of a \emph{finite} number of points. We call these points the \textit{vertices} $V(\mathcal{P})$ of polytope $\mathcal{P}$. Additionally, the vertices of a rational polytope always have rational coordinates and encoding length $\poly(\varphi)$ \citep[Lemma 6.2.4]{Grotschel1993:Geometric}.

Since we are interested in constructing algorithms that perform exact computations, any discussion of non-rational numbers is not relevant. Thus, from now on, every time we deal with a polytope we will mean a rational polytope.

In our algorithm, we will also make use of the concept of \emph{conic hull}, which is introduced next.

\begin{definition}[Conic hull of a convex set]
    The conic hull of a convex set $\mathcal{X}$ is defined as
    \[
        \bbR_+ \mathcal{X} = \{t \cdot \vec x \mid t \geq 0, \vec x \in \mathcal{X} \}.
    \]
\end{definition}

It is known that when $\mathcal{X}$ is a rational polyhedron, its conic hull is also a rational polyhedron.

Moving on to some more specialized results, we introduce the following Lemma. Since in \cref{sec:eah_exact} we will be working with sets of the form described in \cref{lem:polytope_linear_program}, it will be a useful to know that they are in fact polytopes and that their facet-complexity is properly bounded.

\begin{restatable}{lemma}{lempolylinprog}\label{lem:polytope_linear_program}
    Let $\mat{A}$ be a matrix and
    \[
        \mathcal{P} = \left\{ \vec{u} \in \mathcal{U} \bigm\vert \max_{\vec q \in \mathcal{Q}} \vec{q}^\top \mat{A} \vec{u} \leq c
        \right\}.
    \]
    If the following conditions hold
    \begin{itemize}
        \item $\mathcal{U}$ be a rational polyhedron with facet-complexity at most $\varphi$,

        \item $\mathcal{Q}$ be the convex hull of a set of finitely many points $V(\mathcal{Q}) = \{\hat{\vec q}_1, \dots, \hat{\vec q}_K\}$,

        \item the inequality $(\hat{\vec{q}}^\top \mat{A}) \vec{u} \leq c$ has encoding length at most $\varphi$ for all vertices $\hat{\vec{q}} \in V(\mathcal{Q})$.
    \end{itemize}
    Then the set $\mathcal{P}$ is a rational polytope with facet-complexity at most $\varphi$.
\end{restatable}

We also mention a result in the spirit of the classical Farkas lemma. We will use this result in \cref{sec:eah_exact}. A proof is given in \cref{app:proofs}.

\begin{restatable}[Generalized Farkas lemma]{lemma}{lemfarkas}\label{lem:farkas}
    Let $\mathcal{X} \subset \bbR^M, \mathcal{Y} \subset \bbR^N$ be convex compact sets. Then exactly one of the following two statements is true.
    \begin{enumerate}
        \item There exists $\vec{x} \in \mathcal{X}$ such that $\displaystyle \min_{\vec{y} \in \mathcal{Y}} \vec{x}^\top \mat{A} \vec{y} \geq 0$.

        \item There exists $\vec{y} \in \bbR_+ \mathcal{Y}$ such that $\displaystyle \max_{\vec{x} \in \mathcal{X}} \vec{x}^\top \mat{A} \vec{y} \leq -1$.
    \end{enumerate}
\end{restatable}

\subsection{Game theory definitions}

We begin by defining polyhedral games, following \citet{Gordon08:No}. But first, we need to define multi-linear functions.

\begin{definition}[Multi-linear function]
    Let $V_1, \dots, V_n$ be vector spaces. A function $f : V_1 \times \dots \times V_n \to \bbR$ is said to be multi-linear if for each $p \in [n]$ and fixed $\vec{v}_{-p} \in V_{-p}$ the function $f(\vec{v}_p, \vec{v}_{-p})$ is linear in $\vec{v}_p \in V_p$. In other words, if $\nabla f(\vec{v}_{-p})$ is the gradient of $f(\vec{v}_p, \vec{v}_{-p})$ with respect to $\vec{v}_p$ when $\vec{v}_{-p}$ is fixed, then $f(\vec{v}_p, \vec{v}_{-p}) = \vec{v}_p \cdot \nabla f(\vec{v}_{-p})$.
\end{definition}

\begin{definition}[Polyhedral game]\label{defn:polyhedral_game}
    In a polyhedral game with $n$ players, every player $p \in [n]$ has a polytope\footnote{Note that despite their name, polyhedral games have strategy sets that are  \textit{polytopes}, that is, \emph{bounded} polyhedra.} strategy set $\mathcal{A}_p \subset \bbR^{d_p}$ and a multi-linear utility function $u_p : \mathcal{A}_1 \times \dots \times \mathcal{A}_n \to \bbR$
\end{definition}

Some notable examples of polyhedral games are: normal-form games, where every player has a probability simplex as their strategy set, and extensive-form games, where the strategy sets of the players are the sets of sequence-form strategies \citep{Romanovskii62:Reduction, Koller96:Efficient, vonStengel96:Efficient}. We will refer to the encoding length of the game as the \emph{size of the game}. In games of interest this is usually much smaller than holding the full utility function; for example, extensive-form games are encoded using a game tree and different classes of normal-form games can have other succinct descriptions \citep{Papadimitriou2008:Computing}.

A sub-class of polyhedral games that will be particularly useful in our paper is that of bilinear zero-sum games, which is defined below.

\begin{definition}[Bilinear zero-sum game]
    Let $\mathcal{X} \subset \bbR^M$, $\mathcal{Y} \subset \bbR^N$ be two rational polytopes. A bilinear zero-sum game is a game between two players with strategy sets $\mathcal{X}$ and $\mathcal{Y}$ such that the utility of the $\mathcal{X}$-player is
    \[
        u_1(\vec x, \vec y) = \vec{x}^\top \mat{A} \vec{y},
    \]
    for some $\mat{A} \in \bbQ^{M \times N}$,
    and the utility of the $\mathcal{Y}$-player is $u_2(\vec x, \vec y) = -u_1(\vec x, \vec y)$
\end{definition}

We can now define the notion of a $\Phi$-equilibrium, which generalizes the correlated equilibrium for arbitrary $n$-player polyhedral games and sets of strategy transformations $\Phi$. Before we do that, we first need to define the \emph{corner game} $\Gamma(G)$ of a polyhedral game $G$, following \citet{Gordon08:No, Marks2008:No}. This is the game that arises if we let the action sets of every player $p$ be equal to the vertices $V(\mathcal{A}_p)$ of the polytope strategy set of that player. Note that since $\mathcal{A}_p$ is a polytope, it will have a finite number of vertices. The utilities of this game for a player $p \in [n]$ and pure strategy profile $\vec s \in V(\mathcal{A}_1) \times \dots \times V(\mathcal{A}_n)$ are simply given by $u_p(\vec s)$. In this paper, we will denote the vertices of every strategy set in a polyhedral game as $\Pure_p = V(\mathcal{A}_p)$. We are now ready to define the $\Phi$-equilibrium.

\begin{definition}[$\Phi$-equilibrium]\label{defn:phi_equilibrium}
    Let $G$ be a polyhedral game of $n$ players and $\Phi_p$ be a set of strategy transformations $\phi_p : \mathcal{A}_p \to \mathcal{A}_p$ for each player $p \in [n]$. A $\{\Phi_p\}$-equilibrium for $G$ is a joint distribution $\mu \in \Delta(\Pure_1 \times \dots \times \Pure_n)$ on the pure strategy profiles of $\Gamma(G)$, such that for every player $p \in [n]$ and deviation $\phi \in \Phi_p$ it holds
    \[
        \Ex_{\vec s \sim \mu} [u_p(\vec s)] \geq \Ex_{\vec s \sim \mu} [u_p(\phi(\vec{s}_p), \vec{s}_{-p})].
    \]
\end{definition}

\section{A simple framework for computing equilibria in bilinear zero-sum games using good-enough-response (GER) oracles}\label{sec:eah_exact}

We begin by introducing a simple algorithmic framework (\cref{thm:exact_bilinear_framework}) for computing min-max equilibria in bilinear zero-sum games. As mentioned before, it relies on the idea of good-enough-responses. The motivation behind this is that sometimes a best-response oracle is not known, or even NP-hard to construct (as we prove in \cref{thm:efg_br_hard}). On the contrary, good-enough-responses might be a readily available primitive. For example, we will see in \cref{sec:exact_phi} that a good-enough-response oracle materializes through the use of fixed-point oracles for transformations $\phi \in \Phi$ and this enables us to devise polynomial time algorithms for computing exact $\Phi$-equilibria in polyhedral games.

Let us assume that we have a bilinear zero-sum game $\mathcal{G}(\mathcal{X}, \mathcal{Y}, \mat{A})$, where the strategy sets $\mathcal{X} \subset \bbR^M, \mathcal{Y} \subset \bbR^N$ are rational polytopes. We typically assume that $M \gg N$. Additionally, let
\[
    \texttt{OPT} = \max_{\vec x \in \mathcal{X}} \min_{\vec y \in \mathcal{Y}} \vec{x}^\top \mat{A} \vec{y},
\]
be the value of the game at equilibrium, which is known to us. \textit{In the rest of the paper we assume that $\texttt{OPT} = 0$.} This is without loss of generality because otherwise, it is possible to create a new game with this property by augmenting the vectors $\vec{x}, \vec{y}$ with an extra dimension as follows:
\[
    \left[
    \begin{array}{cc}
        \vec{x}^\top & 1
    \end{array}
    \right]
    \left[
    \begin{array}{cc}
        \mat{A} & \vec{0} \\ 
        \vec{0}^\top & -\texttt{OPT} \\ 
    \end{array}
    \right]
    \left[
    \begin{array}{c}
        \vec{y} \\
        1
    \end{array}
    \right] = \vec{x}^\top \mat{A} \vec{y} - \texttt{OPT}.
\]

The framework we present in this section is a formalization of the following observation: using the minimax theorem \citep{vNeumann1928}, we can see that the below statement
\begin{displayquote}
    (S1) Given any $\vec y \in \mathcal{Y}$ we can find some $\vec x = \vec x (\vec y) \in \mathcal{X}$ such that $\vec{x}^\top \mat{A} \vec{y} \geq 0$.
\end{displayquote}
implies the following
\begin{displayquote}
    (S2) There exists $\vec{x}^* \in \mathcal{X}$ such that $(\vec{x}^*)^\top \mat{A} \vec{y} \geq 0$ for all $\vec{y} \in \mathcal{Y}$.
\end{displayquote}

\noindent This is because, the first statement (S1) is equivalent to
\[
    \min_{\vec y} \max_{\vec x} \vec{x}^\top \mat{A} \vec{y} \geq 0,
\]
while the second statement (S2) is equivalent to
\[
    \max_{\vec x} \min_{\vec y} \vec{x}^\top \mat{A} \vec{y} \geq 0.
\]

We are interested in the following question; \emph{``Is there an efficient algorithm that when given access to an oracle for (S1), it constructs a solution $\vec{x}^*$ for (S2) represented as a mixture of a small number of oracle responses?"}.

\subsection{Good-Enough-Response (GER) oracle}

We begin by formally defining the oracle we presented previously, which we coin a Good-Enough-Response (GER) oracle. It is defined as follows:

\begin{center}
\vspace{5mm}
\begin{mdframed}
    \texttt{GER($\vec{y}$):}\\
    \hspace*{5mm}\texttt{return $(\vec{x}, \vec{x}^\top \mat{A}) \in \mathcal{X} \times \bbQ^N$ s.t. $\vec{x}^\top \mat{A} \vec{y} \geq \texttt{OPT} = 0$}
\end{mdframed}
\vspace{5mm}
\end{center}

\noindent where $\vec y \in \mathcal{Y} \subset \bbR^N$, and $\texttt{OPT} = 0$ as was discussed earlier. Note that this is not a best-response oracle, because it does not return an $\vec{x} \in \mathcal{X}$ that maximizes the utility of the max-player. Rather, it suffices to return a ``good enough response'', hence the name.

In fact, our algorithms will often need to query a GER oracle for $\vec{y}' \in \bbR_+ \mathcal{Y}$ and not just for vectors in $\mathcal{Y}$. This however is not a problem because it suffices to find any $\vec{y} = \vec{y}' / \alpha$ for some $\alpha > 0$ and $\vec{y} \in \mathcal{Y}$ and then query $\texttt{GER}(\vec{y})$ instead. To find such a $\vec{y}$ efficiently we can again, without loss of generality, assume that all vectors $\vec{y} \in \mathcal{Y}$ are augmented with an extra dimension (call it $\vec{y}[\emptyseq]$) such that $\vec{y}[\emptyseq] = 1$ for all $\vec{y} \in \mathcal{Y}$. Then we can find the desired scaling factor immediately because $\vec{y}'[\emptyseq] = \alpha$ if and only if $\vec{y}' = \alpha \vec{y}$ for $\vec{y} \in \mathcal{Y}$.

In addition to a good-enough-response oracle, our algorithm also requires a separation oracle $\texttt{SEP}_{\mathcal{Y}}$ for the polytope $\mathcal{Y}$, which can be easily converted to a separation oracle for $\bbR_+ \mathcal{Y}$ by the same ``augmenting'' argument as before. Combining these two, we can make the final separation oracle (\cref{algo:sep}) that is needed to execute the ellipsoid method on \eqref{eq:dual_cp}, as presented later. Specifically, if $\vec{y} \notin \bbR_+ \mathcal{Y}$ then we simply return a separating hyperplane via $\texttt{SEP}_{\bbR_+ \mathcal{Y}}$, else we return a good-enough-response from \texttt{GER}.

\begin{algorithm}[ht]
    \caption{Separation oracle for the ellipsoid method on \eqref{eq:dual_cp}}\label{algo:sep}
    \KwIn{Separation oracle $\texttt{SEP}_{\bbR_+ \mathcal{Y}}$ for $\bbR_+ \mathcal{Y}$, and a good-enough-response oracle \texttt{GER}.}
    \KwOut{A separating hyperplane $\vec{c}$ for $\vec y$ in \eqref{eq:dual_cp}, and a corresponding vector $\vec x$ from \texttt{GER}, if it exists.}
    
    \eIf{$\texttt{SEP}_{\bbR_+ \mathcal{Y}}$ deems that $\vec y$ is in $\bbR_+ \mathcal{Y}$}{
        Set $(\vec{x}, \vec{c})$ to the output $(\vec{x}, \mat{A}^\top \vec{x})$ of $\texttt{GER}(\vec{y})$\;
    }{
        Set $\vec c$ to the separating hyperplane output by $\texttt{SEP}_{\bbR_+ \mathcal{Y}}$\;
        $\vec{x} = \emptyseq$\;
    }
\end{algorithm}

\subsection{The framework}

Our goal is to compute an $\vec x \in \mathcal{X}$ that is an optimal (min-max) strategy for the max-player. Equivalently, we seek to find a solution to the following linear program:
\[
    \tag{$P$}\label{eq:primal_cp}\text{find}\ &\vec{x}\\
    \text{s.t.}\ &\min_{\vec y \in \mathcal{Y}} \vec{x}^\top \mat{A} \vec{y} \geq 0\\
        &\vec{x} \in \mathcal{X}
\]

This is an LP with $M$ variables which is typically assumed to be much greater (even super-exponentially greater) than $N$. When faced with this situation, one might want to attempt to directly solve the dual of \eqref{eq:primal_cp} using the ellipsoid method. However, this would require a proper separation oracle for the dual problem, which corresponds to a linear optimization oracle, or at least a best-response oracle. But as we explained, the oracle access we have is weaker.

Instead, we focus on the below linear program. Note that for any $\vec{y} \in \bbR_+ \mathcal{Y}$, \texttt{GER($\vec{y}$)} should always return an $\vec{x} \in \mathcal{X}$ such that $(\vec{x}^\top \mat{A}) \vec{y} \geq 0$, which is a violated constraint of the LP. Thus, \texttt{GER} together with a separation oracle for $\bbR_+ \mathcal{Y}$ (combined as in \cref{algo:sep}) can be used as a separation oracle for this LP.
\[
    \tag{$D$}\label{eq:dual_cp} \text{find}\ &\vec y\\
    \text{s.t.}\ &\max_{\vec x \in \mathcal{X}} \vec{x}^\top \mat{A} \vec{y} \leq -1\\
        &\vec{y} \in \bbR_+ \mathcal{Y}
\]
By \cref{lem:farkas} and the fact that \eqref{eq:primal_cp} is feasible, it immediately follows that \eqref{eq:dual_cp} must be infeasible. Despite the infeasibility, and following the ``Against Hope'' step of \citet{Papadimitriou2008:Computing}, we execute the ellipsoid method on \eqref{eq:dual_cp} using \cref{algo:sep} as a separation oracle. The ellipsoid method will run for a number $L = \poly(N)$ of steps and then conclude that \eqref{eq:dual_cp} is infeasible. Let $\vec{x}_1, \dots, \vec{x}_L$ be the response vectors returned by \texttt{GER} during this process. We now consider a ``compressed" version of the previous LP that only uses vectors $\vec{x}$ from the convex hull $\co\{\vec{x}_k\}$ of these responses.
\[
    \tag{$D'$}\label{eq:dual_compressed_cp} \text{find}\ &\vec y\\
    \text{s.t.}\ &\max_{\vec x \in \co\{\vec{x}_k\}} \vec{x}^\top \mat{A} \vec{y} \leq -1\\
        &\vec{y} \in \bbR_+ \mathcal{Y}
\]
We argue that this LP must also be infeasible; the ellipsoid method is a deterministic algorithm and if we execute it on \eqref{eq:dual_compressed_cp} it will go through the same sequence of candidate points $\vec{y}_k$, to which we can respond with the same sequence of separating hyperplanes as before. These hyperplanes will still be valid for \eqref{eq:dual_compressed_cp} because all of the response vectors we used previously exist in $\co\{\vec{x}_k\}$.

Now, using \cref{lem:farkas} and the fact that \eqref{eq:dual_compressed_cp} is infeasible, it follows that the LP shown below must be feasible.
\[
    \tag{$P'$}\label{eq:primal_compressed_cp}\text{find}\ &\vec{x}\\
    \text{s.t.}\ &\min_{\vec y \in \mathcal{Y}} \vec{x}^\top \mat{A} \vec{y} \geq 0\\
        &\vec{x} \in \co\{\vec{x}_k\}
\]
This is a ``compressed'' version of \eqref{eq:primal_cp}, because now every vector $\vec{x} \in \co\{\vec{x}_k\}$ can be represented as a vector of size $L$ that corresponds to a convex combination of the response vectors $\vec{x}_1, \dots, \vec{x}_L$.
Finally, since \eqref{eq:primal_compressed_cp}  is an LP with only a polynomial number of variables, we can solve it in polynomial time using any suitable LP solver (such as the ellipsoid method again).
This will clearly be a valid feasible solution for our initial LP \eqref{eq:primal_cp}, because $\co\{\vec{x}_k\} \subset \mathcal{X}$.

The full algorithm is shown below, in \cref{alg:eah}. Note that in reality we only need to use the LPs \eqref{eq:dual_cp} and \eqref{eq:primal_compressed_cp}. The rest were used as intermediate steps for the presentation of the algorithm.

\begin{algorithm}[ht]
    \SetAlgoNoLine
    \caption{Ellipsoid Against Hope for bilinear zero-sum games}\label{alg:eah}
	\KwIn{Separation oracle $\texttt{SEP}_{\bbR_+ \mathcal{Y}}$ for $\bbR_+ \mathcal{Y}$, and a good-enough-response oracle \texttt{GER}.}
	\KwOut{A sparse solution $\vec{x}^*$ of \eqref{eq:primal_cp} represented as a mixture of \texttt{GER} oracle responses.}

    Execute the ellipsoid method on \eqref{eq:dual_cp}, using \cref{algo:sep} as a separation oracle\; %
    Create \eqref{eq:primal_compressed_cp} using the response vectors and compute a feasible solution $\vec{x}^*$\;
\end{algorithm}

\begin{theorem}\label{thm:exact_bilinear_framework}
    If the following hold
    \begin{enumerate}
        \item $\mathcal{X} \subset \bbR^M, \mathcal{Y} \subset \bbR^N$ are rational polytopes,

        \item we have access to a separation oracle $\texttt{SEP}_{\mathcal{Y}}$ for $\mathcal{Y}$ and a good-enough-response oracle \texttt{GER},
        
        \item the facet-complexity of $\mathcal{Y}$ is at most $\varphi$,
        
        \item the encoding length of $\vec{x}^\top \mat{A}$ is at most $\varphi$ for all \texttt{GER} oracle responses and all vertices of $\mathcal{X}$,
    \end{enumerate}
    then \cref{alg:eah} runs in $\poly(N, \varphi)$ time, performs $L = \poly(N, \varphi)$ oracle calls, and computes an exact solution $\vec{x}^*$ of \eqref{eq:primal_cp} that is a mixture of at most $N$ oracle responses. In particular, the encoding length of $\vec{x}^*$ depends polynomially on the encoding length of the \texttt{GER} oracle responses.
\end{theorem}
\begin{proof}
First, we can assume without loss of generality that there exists a dimension $\emptyseq$ in all $\vec{y} \in \mathcal{Y}$ such that $\vec{y}[\emptyseq] = 1$. Otherwise, it is always possible to augment these vectors with an extra dimension before applying the next steps of the algorithm. This allows us to convert the given separation oracle $\texttt{SEP}_\mathcal{Y}$ into a new separation oracle $\texttt{SEP}_{\bbR_+ \mathcal{Y}}$ for $\bbR_+ \mathcal{Y}$ that we can then use to construct the general oracle of \cref{algo:sep}.

The first step of the algorithm is to execute the ellipsoid method on \eqref{eq:dual_cp}. Using the assumptions of the theorem in \cref{lem:polytope_linear_program} with $\mathcal{U} = \bbR_+ \mathcal{Y}$ and $\mathcal{Q} = \mathcal{X}$, it follows that \eqref{eq:dual_cp} is a polytope and has facet-complexity at most $\varphi$. Additionally, by the fact that \eqref{eq:primal_cp} is feasible (it has an equilibrium with $\texttt{OPT} = 0$) and by \cref{lem:farkas} it follows that \eqref{eq:dual_cp} must be infeasible.

To execute the ellipsoid method on \eqref{eq:dual_cp} would then mean, in the language of \citet{Grotschel1993:Geometric}, to solve the Strong Nonemptiness Problem for \eqref{eq:dual_cp} using the strong separation oracle of \cref{algo:sep}. To this end, we use the algorithm from Theorem 6.4.1 of \citet{Grotschel1993:Geometric}. This algorithm works for any polyhedron, even if it is not bounded or full-dimensional, as might be the case here. To do that, it might execute the central-cut ellipsoid method more than once, but never more than $N$ times. In our case, we already know that \eqref{eq:dual_cp} is infeasible and thus, the algorithm terminates after $N$ executions of the central-cut ellipsoid and concludes that \eqref{eq:dual_cp} is infeasible.

Since the central-cut ellipsoid method is an oracle-polynomial algorithm that is executed $N$ times in polyhedra of facet-complexity at most $\varphi$, the whole process runs in polynomial time and performs a polynomial number of separation oracle calls. To calculate the exact number $L$ of oracle calls, we note that the algorithm in \citet[Theorem 6.4.1]{Grotschel1993:Geometric} initializes the central-cut ellipsoid method with
\[
    R = 2^{O(N^2 \varphi)} \ \text{ and }\  \epsilon = 2^{-O(N^5 \varphi)},
\]
while the central-cut method terminates in $O(N \log(1/\epsilon) + N^2 \log R)$ iterations \citep[Theorem 3.2.1]{Grotschel1993:Geometric}. Combining these with the fact that the central-cut ellipsoid method is repeated $N$ times, we get that the number of oracle calls is $L = O(N^7 \varphi)$.

Next, note that \eqref{eq:dual_compressed_cp} is comprised of constraints coming from \texttt{GER} oracle responses, which by \cref{lem:polytope_linear_program} gives that the facet-complexity of \eqref{eq:dual_compressed_cp} must also be at most $\varphi$. By going through the same process as before, the algorithm will reach the same conclusion after executing the central-cut ellipsoid method $N$ times; \eqref{eq:dual_compressed_cp} is infeasible.

Finally, by the infeasibility of \eqref{eq:dual_compressed_cp} and \cref{lem:farkas}, it follows that \eqref{eq:primal_compressed_cp} must be feasible. An equivalent way to express \eqref{eq:primal_compressed_cp} is
\[
    \text{find}\ &\vec{a}\\
    \text{s.t.}\ &\min_{\vec y \in \mathcal{Y}} \vec{a}^\top (\mat{X}^\top \mat{A}) \vec{y} \geq 0\\
    &\vec{a} \in \Delta^L,
\]
where $\Delta^L$ is the $L$-dimensional simplex and $\mat{X} = [\vec{x}_1 \mid \dots \mid \vec{x}_L]$ is a matrix with the \texttt{GER} oracle responses as its columns.

Applying \cref{lem:polytope_linear_program} for $\mathcal{U} = \Delta^L$ and $\mathcal{Q} = \mathcal{Y}$ we conclude that \eqref{eq:primal_compressed_cp} describes a polytope
\[
    \mathcal{P} = \left\{ \vec{a} \in \Delta^L \bigm\vert \min_{\vec{y} \in \mathcal{Y}} \vec{a}^\top (\mat{X}^\top \mat{A}) \vec{y} \geq 0 \right\}
\]
of encoding length at most $L \poly(\varphi)$. This is because, for any vertex $\hat{\vec{y}} \in V(\mathcal{Y})$, the inequality $\vec{a}^\top (\mat{X}^\top \mat{A}) \hat{\vec{y}} \geq 0$ has $L$ coefficients, each of which having encoding length $\poly(\varphi)$.
This can be solved in polynomial time by any known linear programming method. Even better, it is possible to compute a basic feasible solution of this LP, which will have at most $N$ non-zero entries and thus the final solution $\vec{x}^* = \mat{X} \vec{a}$ will be a mixture of at most $N$ oracle responses.
\end{proof}

Note that since we have assumed that $M \gg N$, it would not make sense for the final solution $\vec{x}^*$ to have encoding length $\poly(M)$, as this would invalidate the whole algorithm. In order for the solution to make sense, the \texttt{GER} oracle must only give responses with low encoding length. This is exactly the case in \cref{sec:exact_phi}, where $M$ is a doubly-exponential quantity in the size of the problem, while the \texttt{GER} responses are vectors with only one non-zero entry.

\section{Computing linear $\Phi$-equilibria in polynomial time}\label{sec:exact_phi}

We have seen in \cref{sec:eah_exact} how one can compute exact min-max equilibria using good-enough-response (\texttt{GER}) oracles. Now it is time to apply this machinery in the problem of computing \emph{exact} $\Phi$-equilibria in polyhedral games. Crucially, the factor that enables us to utilize the framework of \cref{sec:eah_exact} is the existence of an efficient \texttt{GER} oracle, which effectively boils down to constructing a product distribution consisting of fixed-points for the strategies of every player of the game.

Let $G$ be any polyhedral game (\cref{defn:polyhedral_game}) with $n$ players and strategy sets $\mathcal{A}_p \subset \bbR^{d_p}$ for $p \in [n]$. In this section we apply the framework we developed previously to construct an algorithm that computes an exact $\Phi$-equilibrium of $G$ in polynomial time when $\Phi$ is a polytope containing valid linear transformations from polyhedral strategies to polyhedral strategies. Notable examples of sets with these properties are the trigger deviations used for EFCE \citep{Farina22:Simple}, and the linear-swap deviations used for LCE \citep{Farina2023:Polynomial} in extensive-form games.

The general idea of our construction is that of the existence proof by \citet{Hart1989:Existence} that casts the problem of $\Phi$-equilibrium computation as one of computing a min-max equilibrium in a two-player zero-sum meta-game between a ``Correlator'', who acts upon the simplex of all pure strategy profiles, and a ``Deviator'', whose actions correspond to deviations for every player. We call this a \emph{Correlator-Deviator game}.

To make this idea applicable to polyhedral games, we generalize it as follows. We define a bilinear zero-sum meta-game with strategy sets $\mathcal{X}, \mathcal{Y}$ for the two players, where $\mathcal{X}$ is the set of all joint distributions over strategy profiles, $\cX = \Delta(\Pure_1 \times \dots \times \Pure_n)$ (hence, a polytope) and $\mathcal{Y}$ is the Cartesian product of $\Phi_p$ for all players $p$, $\mathcal{Y} = \Phi_1 \times \dots \times \Phi_n$, which is a convex set -- and in our case, a polytope.

We remark here that linear transformations $\phi_p$ can be represented using a matrix $\mat{B}_p$ such that $\phi_p(\vec{x}_p) = \mat{B}_p \vec{x}_p$. Thus, when we say that $\Phi_p$ is a polytope, it means that there exists a system of inequalities that can describe the entries of the corresponding matrix $\mat{B}_p$ for every $\phi_p \in \Phi_p$. For notational convenience, we will interchangeably use $\Phi_p$ to denote either the set of transformation functions, or a polytope describing the vectors (flattened $\mat{B}_p$ matrices) that correspond to transformations. In any event, it should not matter which of the two representations we have, because they are completely equivalent.

The utility matrix $\mat{U}$ of this meta-game is shown below. Specifically, it has one row for each pure strategy profile $\vec s \in \Pure_1 \times \dots \times \Pure_n$, and one column for each tuple $j = (p, a, b)$, where $a, b \in [d_p]$ are used as indices over strategy vectors $\vec{s}_p \in \mathcal{A}_p$. Finally, the final expression always needs to have a quantity ($\sum_p \Ex_{s \sim \vec x} [u_p(s)]$) that is independent of the value of $\vec{y}$. To achieve this we can use a trick similar to the one used to make $\texttt{OPT} = 0$ in \cref{sec:eah_exact} by augmenting vectors $\vec{y} \in \mathcal{Y}$ with an extra dimension $\emptyseq$ such that $\vec{y}[\emptyseq] = 1$ always holds. Then we have\footnote{As an aside, we are slightly abusing the notation here because $u_p$ is defined as a function with domain $\mathcal{A}_1 \times \dots \times \mathcal{A}_n$ and it is not guaranteed that $\vone_b \in \mathcal{A}_p$. However, this is not a problem due to the multi-linearity of the utilities that basically makes $u_p(\cdot, \vec{s}_{-p})$ a vector of size $d_p$.}
\[
    \mat{U}_{s j} = \mleft\{ \begin{array}{ll}
        \sum_p u_p(s), & j = \emptyseq \\
        -\vec{s}_p[a] u_p(\vone_b, \vec{s}_{-p}), & \text{otherwise}
    \end{array} \mright.
\]
where $\vone_b$ denotes the vector having all $0$, apart from index $b$, which is $1$.
Note that the number of rows of $\mat{U}$ might be doubly-exponential (exponential both in the number of players and the dimension of the polyhedral strategies), which is in contrast to the original Ellipsoid Against Hope algorithm that only allowed a number of rows exponential in the number of players.

\begin{lemma}
    Let $G$ be a polyhedral game with pure strategy set $\Pure_p$ for every player $p \in [n]$. Additionally, let $\Phi_p$ be a set of linear transformations for every $p \in [n]$.
    If $\vec{x} \in \mathcal{X} = \Delta(\Pure_1 \times \dots \times \Pure_n)$ and $\vec{y} = (\phi_1, \dots, \phi_n) \in \mathcal{Y} = \Phi_1 \times \dots \times \Phi_n$ then
    \[
        \vec{x}^\top \mat{U} \vec{y} = \sum_{p} \Ex_{s \sim \vec{x}}[u_p(s) - u_p(\phi_p(\vec{s}_p), \vec{s}_{-p})].
    \]
\end{lemma}
\begin{proof}
    As we discussed, each linear transformation $\phi_p$ can be viewed as a $d_p \times d_p$ transformation matrix $\mat{B}_p$. We denote the matrix entries with $\mat{B}_p[b, a]$. In particular, if $\vec{s}_p' = \phi_p(\vec{s}_p)$, we have $\vec{s}_p'[b] = \sum_a \mat{B}_p[b, a] \vec{s}_p[a]$. Then for any $\vec{x} \in \cX$ and $\vec{y} = (\phi_1, \dots, \phi_n) \in \mathcal{Y}$ we have
    \[
        \vec{x}^\top \mat{U} \vec{y} &= \sum_{s} x_s \sum_{p} \left( u_p(s) - \sum_{a \in [d_p]} \sum_{b \in [d_p]} \mat{B}_p[b, a] \vec{s}_p[a] u_p(\vone_b, \vec{s}_{-p}) \right)\\
        &= \sum_{p} \sum_{s} x_s \left( u_p(s) - u_p\left(\sum_{b \in [d_p]} \vone_b \sum_{a \in [d_p]} \mat{B}_p[b, a] \vec{s}_p[a], \vec{s}_{-p} \right) \right)\\
        &= \sum_{p} \sum_{s} x_s \left( u_p(s) - u_p\left(\sum_{b \in [d_p]} \vone_b \vec{s}_p'[b], \vec{s}_{-p} \right) \right)\\
        &= \sum_{p} \sum_{s} x_s \left( u_p(s) - u_p(\phi_p(\vec{s}_p), \vec{s}_{-p}) \right)\\
        &= \sum_{p} \Ex_{s \sim \vec{x}}[u_p(s) - u_p(\phi_p(\vec{s}_p), \vec{s}_{-p})],
    \]
    where in the second equality we have used the multi-linearity of the utilities.
\end{proof}

It is now evident that our goal is to compute a joint distribution that is a solution to the following linear program
\[
    \numberthis{eq:phi_meta_game} \text{find}\ & \vec{x}\\
    \text{s.t.}\ &\min_{\vec{y} \in \mathcal{Y}} \vec{x}^\top \mat{U} \vec{y} \geq 0 \\
        &\vec{x} \in \cX
\]
Observe that this is slightly different from the required non-negativity in \cref{defn:phi_equilibrium} in that it suffices for the minimum of the \emph{sum} of expectations to be non-negative. In general this might not necessarily mean that the individual expectations are non-negative, but we can assume without loss of generality that the identity transformation is always a valid transformation, because otherwise we can replace each $\Phi_p$ with $\co\{\Phi_p \cup \{\mat{I}\}\}$ which remains a rational polytope and admits a separation oracle when $\Phi_p$ has a separation oracle. Then, we have enough degrees of freedom to minimize over all individual expectations. Alternatively, we could have also defined $\mathcal{Y}$ to be
\[
    \mathcal{Y} = \co\left\{ 
        \begin{pmatrix}
            \Phi_1 \\
            0 \\
            \vdots \\
            0
        \end{pmatrix} \cup
        \begin{pmatrix}
            0 \\
            \Phi_2 \\
            \vdots \\
            0
        \end{pmatrix} \cup \dots \cup
        \begin{pmatrix}
            0 \\
            0 \\
            \vdots \\
            \Phi_n
        \end{pmatrix}
    \right\}
\]
to explicitly create constraints for each $\phi_p \in \Phi_p$ for all $p \in [n]$ while keeping everything else $0$. In the analysis that follows we focus on the first case where $\mathcal{Y}$ is the Cartesian product of the sets of transformations, but the same arguments hold in  both cases.

The LP \eqref{eq:phi_meta_game} respects exactly the structure of \eqref{eq:primal_cp} that our min-max framework can handle. The only component that we need to get a polynomial-time algorithm is to have an efficient good-enough-response oracle \texttt{GER}. Specifically, for any valid $\vec{y} \in \mathcal{Y}$, we need to respond with an $\vec{x}$ such that $\vec{x}^\top \mat{U} \vec{y} \geq 0$. Similar to the original Ellipsoid Against Hope algorithm \citep{Papadimitriou2008:Computing} and based on the observation by \citet{Hart1989:Existence}, the important insight that allows us to construct an efficient oracle and uncover sparse solutions is that we can always find such an $\vec{x}$ that is a product distribution. Note that since we can use polyhedral strategies to represent the marginals, it follows that the product distribution $\vec{x}$ requires space that is only linear in the game size.

Before we present the Lemma for the \texttt{GER} oracle, we give an important assumption regarding a property that a game necessarily has to have to enable an efficient implementation of the \texttt{GER} oracle.

\begin{assumption}[Polynomial utility gradient property]\label{asmpt:utility_gradient}
    Given a product distribution $\vec{x} \in \Delta(\Pure_1 \times \dots \times \Pure_n)$, it is possible to compute the value of
    \[
        \vec{g}_p(\vec{x}_{-p}) = \Ex_{\vec{s}_{-p} \sim \vec{x}_{-p}}[\nabla u_p (\vec{s}_{-p})]
    \]
    for all players $p \in [n]$ in polynomial time in the encoding length of $\vec{x}$ and the size of the game.
\end{assumption}

This assumption generalizes the polynomial expectation property defined in \citet{Papadimitriou2008:Computing} to more general, polyhedral games. In particular, if we have a normal-form game, the polynomial expectation property amounts to computing $\vec{g}_p(\vec{x}_{-p}) \cdot \vec{x}_p$ for a product distribution $\vec{x}$.

\begin{remark}\label{rmrk:polynomial_type}
    \citet{Papadimitriou2008:Computing} also defined a second property that is required for efficient computation, called the ``polynomial type'' property. Even though our algorithm does \emph{not} require this property, a variant of it is implicit in the fact that the complexity of the algorithm depends polynomially in the number of players and the dimension of every player's strategy set $\mathcal{A}_p$. However, this relaxation is what allows our algorithm to handle much broader classes of games, such as the extensive-form games that do not have the polynomial type property.
\end{remark}

Now we are ready to present the good-enough-response oracle that will allow us to develop an efficient algorithm. As a general backbone, this Lemma follows the constructive proof that \citet{Papadimitriou2008:Computing} did for the CE existence result of \citet{Hart1989:Existence} and, crucially, it produces pure strategies using the idea of \citet{Jiang2015:Polynomial}. 

\begin{lemma}[\texttt{GER} oracle for $\Phi$-equilibria]\label{lem:ger_phi_equil}
    For every $\vec{y} \in \mathcal{Y} = \Phi_1 \times \dots \times \Phi_n$ there exists a pure strategy profile $s \in \Pure_1 \times \dots \times \Pure_n$ such that $\vone_s^\top \mat{U} \vec{y} \geq 0$. Furthermore, such a strategy profile alongside with the vector $\vone_s^\top \mat{U}$ can be computed efficiently, provided that the game satisfies the polynomial utility gradient property (\cref{asmpt:utility_gradient}) and there exists a polynomial-time separation oracle for every $\mathcal{A}_p$.
\end{lemma}

\begin{proof}
As we have discussed, we can denote all linear transformations $\phi_p$ using a matrix $\mat{B}_p$ such that $\phi_p(\vec{x}_p) = \mat{B}_p \vec{x}_p$.

First note that there always exists a fixed-point of any linear strategy transformation $\phi_p$; this follows from Brouwer's fixed-point theorem and the fact that these transformations are continuous maps of a compact convex set $\mathcal{A}_p$ to itself. Additionally, since the transformations are linear we can always efficiently compute a fixed-point of any transformation by solving the following LP:
\[
    \text{find}\ & \vec{x}_p\\
    \text{s.t.}\ &\mat{B}_p \vec{x}_p = \vec{x}_p \\
    &\vec{x}_p \in \mathcal{A}_p
\]
that can be solved in polynomial time using the ellipsoid method with the given separation oracle for $\mathcal{A}_p$.

Next, let us restrict our attention only to product distributions $\vec{x} \in \Delta(\Pure_1 \times \dots \times \Pure_n)$. In this case it will be $x_s = x_{-p}(\vec{s}_{-p}) x_p(\vec{s}_p)$ for all pure strategy profiles $s$, which gives
\[
    \vec{x}^\top \mat{U} \vec{y} &= \sum_{p} \sum_{\vec{s}_{-p}} \sum_{\vec{s}_p \in \Pure_p} x_{-p}(\vec{s}_{-p}) x_p(\vec{s}_p) \left[ u_p(\vec{s}_p, \vec{s}_{-p}) - u_p(\phi_p(\vec{s}_p), \vec{s}_{-p})
    \right]\\
    &= \sum_{p} \sum_{\vec{s}_{-p}} x_{-p}(\vec{s}_{-p}) u_p\left( \left[ \vec{x}_p - \phi_p(\vec{x}_p) \right], \vec{s}_{-p} \right)\\
    &= \sum_{p} \vec{g}_p(\vec{x}_{-p}) \cdot \left[ \vec{x}_p - \phi_p(\vec{x}_p) \right] \numberthis{eq:product_utility_grad},
\]
where $\vec{x}_p \in \mathcal{A}_p$ is the marginal distribution for player $p$ represented as a point of the polyhedral strategy set.
In the second equality we have used the multi-linearity of $u_p(\cdot, \vec{s}_{-p})$ and the linearity of the transformations; $\sum_{\vec{s}_p} x_p(\vec{s}_p) \phi_p(\vec{s}_p) = \phi_p(\vec{x}_p)$. It directly follows from the last equality that if we set each marginal distribution equal to the corresponding fixed-point $\vec{x}_p = \phi_p(\vec{x}_p)$, we get a product distribution $\vec{x}$ such that $\vec{x}^\top \mat{U} \vec{y} = 0$.

Now, it remains to find a way to extract the desired pure strategy profile $s$ from this product distribution. We follow a similar procedure to the purification technique used by \citet{Jiang2015:Polynomial}. Similar to their algorithm, we define $\vec{x}_{(p \to \vec{s}_p)}$ for a product distribution $\vec{x}$ to be the product distribution in which player $p$ plays pure action $\vec{s}_p \in \mathcal{A}_p$ and all other players act according to $\vec{x}_{-p}$. Additionally, note that since $\mathcal{A}_p$ is a polytope, it must hold
\[
    \vec{x}_p = \sum_{\vec{s}_p \in V(\mathcal{A}_p)} \lambda_{\vec{s}_p} \vec{s}_p
\]
for some convex combination $\{ \lambda_{\vec{s}_p} \geq 0 \mid \sum_{\vec{s}_p} \lambda_{\vec{s}_p} = 1 \}$.
By the product distribution structure, it is easy to see that for every player $p \in [n]$,
\[
    \numberthis{eq:convex_comb_pure} \vec{x}^\top \mat{U} \vec{y} = \sum_{\vec{s}_p \in V(\mathcal{A}_p)} \left[ \vec{x}_{(p \to \vec{s}_p)}^\top \mat{U} \vec{y} \right] \lambda_{\vec{s}_p}
\]

The algorithm of \citet{Jiang2015:Polynomial} iterates over all players and for each player $p$ they search over all its pure strategies and find one, $\vec{s}_p^*$, for which $\vec{x}_{(p \to \vec{s}_p^*)}^\top \mat{U} \vec{y} \geq 0$. Such a pure strategy must always exist because \eqref{eq:convex_comb_pure} represents a convex combination over all vertices (a.k.a. pure strategies). However, in our case we cannot iterate over all pure strategies for a player because they might be exponentially many.

To make this procedure general for all polyhedral games, we observe that by Carath\'eodory's theorem there must always exist a subset $\left\{ \hat{\vec{s}}_p\^1, \dots, \hat{\vec{s}}_p\^k \right\} \subset V(\mathcal{A}_p)$ of at most $k \leq d_p + 1$ vertices of $\mathcal{A}_p$ that satisfy
\[
    \vec{x}_p = \sum_{i=1}^k \lambda_i \hat{\vec{s}}_p\^i
\]
for some convex combination represented with $\lambda_1, \dots, \lambda_k$. Thus, we can follow the same procedure as before but this time only search over $k$ vertices instead of all (possibly exponentially many) vertices of $\mathcal{A}_p$. The complete algorithm is shown in \cref{alg:purified_ger_oracle}.

This can be implemented in polynomial time because: (a) there exists an algorithmic version of Carath\'eodory's theorem \citet[Theorem 6.5.11]{Grotschel1993:Geometric} that only requires access to a separation oracle for $\mathcal{A}_p$, and (b) \cref{asmpt:utility_gradient} allows us to compute $\vec{x}_{(p \to \vec{s}_p^*)}^\top \mat{U} \vec{y}$ in polynomial time for any product distribution $\vec{x}_{(p \to \vec{s}_p^*)}$, as is evident from \eqref{eq:product_utility_grad}.
\end{proof}

\begin{algorithm}[ht]
    \SetAlgoNoLine
    \caption{Purified \texttt{GER} oracle}\label{alg:purified_ger_oracle}
	\KwIn{Polyhedral game G of $n$ players, $\vec{y} \in \mathcal{Y} = \Phi_1 \times \dots \times \Phi_n$, separation oracles $\texttt{SEP}_{\mathcal{A}_p}$ for all $p \in [n]$.}
	\KwOut{A pure strategy profile $s \in \Pure_1 \times \dots \times \Pure_n$ such that $\vone_s^\top \mat{U} \vec{y} \geq 0$.}

    Compute a product distribution $\vec{x}$ s.t. $\vec{x}^\top \mat{U} \vec{y} = 0$ by finding fixed points $\vec{x}_p = \phi_p(\vec{x}_p)$ for all $p \in [n]$ \;

    \For{$p \in [n]$}{
        Find a set of $k \leq d_p+1$ vertices $\left\{ \hat{\vec{s}}_p\^1, \dots, \hat{\vec{s}}_p\^k \right\} \subset V(\mathcal{A}_p)$  s.t. $\vec{x}_p = \sum_{i=1}^k \lambda_i \hat{\vec{s}}_p\^i$ \;

        Set $\vec{s}_p^*$ to the vertex $\hat{\vec{s}}_p\^i$ that satisfies $\vec{x}_{(p \to \vec{s}_p^*)}^\top \mat{U} \vec{y} \geq 0$\;
        
        Set $\vec{x}$ to be $\vec{x}_{(p \to \vec{s}_p^*)}$ \;
    }
    Finally, $\vec{x}$ must correspond to a pure strategy profile $s$\;
\end{algorithm}

\begin{theorem}\label{thm:phi_equilibria_polyhedral}
    Let $G$ be a polyhedral game (\cref{defn:polyhedral_game}) of $n$ players and $\{\Phi_p\}$ be a collection of polytopes corresponding to sets of linear strategy transformations that map every strategy set $\mathcal{A}_p$ to itself. Additionally, let $N = \sum_p d_p^2$. Assume that
    \begin{itemize}
        \item there exist polynomial-time separation oracles for $\mathcal{A}_p$ and $\Phi_p$,

        \item $G$ satisfies the polynomial utility gradient property (\cref{asmpt:utility_gradient}),
        
        \item $\psi$ is an upper bound on the facet-complexity of every $\mathcal{A}_p$ and $\Phi_p$,

        \item $\log u$ is the maximum encoding length of the utilities of $G$.
    \end{itemize}
    Then there exists an algorithm that computes an exact $\{\Phi_p\}$-equilibrium of $G$ in time $\poly(N, \log u, \psi)$ and performs $\poly(N, \log u, \psi)$ number of calls to all the separation oracles. Additionally, the equilibrium is represented as a convex combination of at most $N$ pure strategy profiles.
\end{theorem}
\begin{proof}
    The set $\mathcal{X} = \Delta(\Pure_1 \times \dots \times \Pure_n)$ is trivially a rational polytope and the set $\mathcal{Y} = \Phi_1 \times \dots \times \Phi_n \subset \bbR^N$ is the Cartesian product of rational polytopes, hence a rational polytope.
    Furthermore, we can directly construct a polynomial-time separation oracle for $\mathcal{Y}$ by calling the separation oracles for each one of the sets $\Phi_p$.
    Additionally, every row of the $\mat{U}$ matrix has $N$ entries, each with encoding length at most $2 \log u$.
    Using the good-enough-response oracle from \cref{lem:ger_phi_equil}, each response $(\vec{x}_k, \vec{x}_k^\top \mat{U}) \in \mathcal{X} \times \bbQ^N$ corresponds to a pair of a pure strategy profile (vertex of $\mathcal{X}$) and a row of $\mat{U}$. Thus, each response has encoding length at most $2 N \log u$. Set $\varphi = \max(2 N \log u, \psi)$.
    
    Now, we can apply \cref{thm:exact_bilinear_framework}, which gives us an algorithm running in $\poly(N, \varphi)$ time and performing $\poly(N, \varphi)$ oracle calls. Combining this with \cref{lem:ger_phi_equil}, it follows that the total time complexity is $\poly(N, \varphi) = \poly(N, \log u, \psi)$. Finally, the optimal solution $\vec{x}^*$ will be comprised of a mixture of $N$ oracle responses. In other words, $\vec{x}^*$ will be an exact $\{\Phi_p\}$-equilibrium for the game $G$ with probability mass on at most $N$ pure strategy profiles.
\end{proof}

As a first application of this framework, we argue that it can be applied to normal-form games that satisfy the polynomial type and the polynomial expectation property, defined in \citet{Papadimitriou2008:Computing}. More precisely, a normal-form game is a polyhedral game where every strategy set is a probability simplex. Additionally, the set $\Phi$ of all linear transformations in normal-form games is that of swap-deviations which is equivalent to the set of all stochastic matrices \citep{Gordon08:No}. Both the sets of strategies and the sets of stochastic matrices can easily be represented as polytopes of bounded facet-complexity having polynomially many constraints. Finally, the polynomial utility gradient property (\cref{asmpt:utility_gradient}) reduces to the polynomial expectation property in normal-form games and the polynomial type property is implicitly satisfied (see \cref{rmrk:polynomial_type}). Thus, all requirements of \cref{thm:phi_equilibria_polyhedral} are satisfied and we have just proven the following corollary.

\begin{corollary}[Exact CE in normal-form games]
    If a normal-form game $G$ has the polynomial type and the polynomial expectation property, defined in \citet{Papadimitriou2008:Computing}, then our algorithm reduces to the Ellipsoid Against Hope (more precisely, the version by \citet{Jiang2015:Polynomial}) and computes an exact correlated equilibrium of $G$.
\end{corollary}

As we have discussed, a very notable example of polyhedral games is that of extensive-form games. Next, we apply \cref{thm:phi_equilibria_polyhedral} to this class of games, and specifically to the set of all linear-swap deviations, recently defined in \citet{Farina2023:Polynomial}. In particular, this set contains all trigger deviations and thus, our algorithm also produces an extensive-form correlated equilibrium (EFCE) in a conceptually simpler manner than in the early work of \citet{Huang2008:Computing}.

\begin{corollary}[Exact LCE computation]
    There exists an algorithm that runs in $\poly(N, \log u)$ time and computes an exact linear-deviation correlated equilibrium (LCE) in an extensive-form game.
\end{corollary}
\begin{proof}
    We apply \cref{thm:phi_equilibria_polyhedral} for the set of linear-swap deviations. Specifically, in \citet[Theorem 3.1]{Farina2023:Polynomial} it is proved that:
    \begin{itemize}
        \item The set $\Phi_\text{LIN}$ of linear-swap deviations for a player $p$ is a rational polytope.
        
        \item This polytope can be described using a polynomial number of equality constraints, which immediately implies the existence of an efficient separation oracle.

        \item Every constraint of the characterization has at most $|\Seqs_p|^2$ coefficients, each belonging to $\{0, 1, -1\}$. Thus, the facet-complexity of $\Phi_\text{LIN}$ must be $\psi = |\Seqs_p|^2$.
    \end{itemize}
    Finally, since the number of non-zero utilities are at most equal to the game tree size, it trivially follows that extensive-form games satisfy the polynomial utility gradient property (\cref{asmpt:utility_gradient}).
    It follows that there exists a polynomial time algorithm for computing LCEs.
\end{proof}

Finally, we prove in \cref{thm:efg_br_hard} that, at least in the case of computing $\Phi$-equilibria in polyhedral games, the use of a good-enough-response over a best-response oracle is not just more elegant, but it is also necessary because constructing a best-response oracle is NP-hard. At the heart of our hardness result lies a reduction from SAT to equilibrium computation in extensive-form games that has also been used in the past to prove the hardness of equilibrium selection for EFCE and LCE \citep{vonStengel2008, Farina2023:Polynomial}. In a sense, constructing a best-response oracle is as hard as the equilibrium selection problem, while constructing a good-enough-response oracle amounts to computing fixed-points of strategy transformation functions. This further highlights the importance of having a framework akin to the one presented in \cref{sec:eah_exact} for designing new algorithms.

\begin{theorem}[Hardness of BR oracle]\label{thm:efg_br_hard}
    It is NP-hard to construct a best-response oracle for the Correlator in the Correlator-Deviator game.
\end{theorem}
\begin{proof}
    A best-response oracle for the Correlator in the Correlator-Deviator game must respond with the optimal $\vec{x} \in \mathcal{X} = \Delta(\Pure_1 \times \dots \times \Pure_n)$ for any given $\vec{y} \in \mathcal{Y} = \Phi_1 \times \dots \times \Phi_n$. More precisely, we have to be able to compute
    \[
        \vec{x}^* = \argmax_{\vec{x} \in \mathcal{X}} \left\{ \vec{x}^\top \mat{U} \vec{y} \right\} = \argmax_{\vec{x} \in \mathcal{X}} \left\{ \sum_{p} \Ex_{s \sim \vec{x}}[u_p(s) - u_p(\phi_p(\vec{s}_p), \vec{s}_{-p})] \right\}
    \]
    for all $\vec{y} \in \mathcal{Y}$.

    To prove that this process is intractable it suffices to find a game and an equilibrium concept such that it is NP-hard to compute $\vec{x}^*$ for at least one $\vec{y} \in \mathcal{Y}$. For the solution concept we choose the coarse-correlated equilibrium, in which the sets $\Phi_p$ consist of all constant (or external) deviations that output a fixed strategy $\phi_p(\vec{x}_p) = \bar{\vec{s}}_p \in \Pure_p$ no matter the input strategy $\vec{x}_p$. For the game, we choose the SAT-game that was also used to prove the hardness of equilibrium selection for EFCE and LCE in extensive-form games \citep{vonStengel2008, Farina2023:Polynomial}. The exact details are not important for our purposes, but we will only use the fact that any SAT instance can be encoded in a 2-player extensive-form game using a polynomial-time reduction. In this game, any pure strategy profile $s \in \Pure_1 \times \dots \times \Pure_n$ with social welfare (sum of players' utilities) equal to $2$ corresponds to a satisfying assignment for the SAT instance, while any other strategy profile has social welfare at most $2(1 - 1/n)$. Thus, there exists a Nash equilibrium (and hence, a CCE) corresponding to a pure strategy profile that has maximum social welfare.

    Before we proceed, we augment the SAT-game by adding an extra decision point for both players at the beginning of the game that asks whether they want to play. If both players respond ``Yes'' then the game continues as normal, otherwise --if at least one player responds ``No''-- the game ends and both players get a $0$ payoff. We denote the pure strategy of the ``No'' answer from player $1$ with $\vec{s}_1^{N}$ and the ``No'' answer from player $2$ with $\vec{s}_2^{N}$.

    Now, to prove the desired result consider a Correlator-Deviator game applied to the computation of a CCE in the above augmented SAT-game. Assume that there exists a polynomial-time best-response oracle for this game that returns a solution of polynomial size in the representation of the game (and hence, the size of the SAT instance). Then, we can use it to best-respond to $\vec{y} = (\phi_1, \phi_2)$ for $\phi_1(\vec{x}_1) = \vec{s}_1^N$ and $\phi_2(\vec{x}_2) = \vec{s}_2^N$. Specifically, we have
    \[
        \vec{x}^* &= \argmax_{\vec{x} \in \mathcal{X}} \left\{ \sum_{p} \Ex_{s \sim \vec{x}}[u_p(s) - u_p(\phi_p(\vec{s}_p), \vec{s}_{-p})] \right\}\\
        &= \argmax_{\vec{x} \in \mathcal{X}} \left\{ \sum_{p} \Ex_{s \sim \vec{x}}[u_p(s)] - \sum_{p} \Ex_{s \sim \vec{x}} [ u_p(\vec{s}_p^N, \vec{s}_{-p})] \right\}\\
        &= \argmax_{\vec{x} \in \mathcal{X}} \left\{ \sum_{p} \Ex_{s \sim \vec{x}}[u_p(s)] \right\}.
    \]
    In other words, the BR oracle returns in this case a distribution $\vec{x}^*$ over pure strategy profiles with maximum social welfare.
    Since the BR oracle computes a polynomially-sized $\vec{x}^*$ in polynomial-time, we can uncover a pure strategy profile of maximum social welfare that corresponds, by construction of the SAT-game, to a satisfying assignment of the SAT instance. We conclude that constructing a best-response oracle in the Correlator-Deviator game corresponding to the compution of a CCE in extensive-form games is NP-hard.
\end{proof}

\section{Discussion and Future Work}

In this paper, we devise a polynomial-time algorithm for computing min-max equilibria in bilinear zero-sum games, by utilizing a good-enough-response oracle. We use this machinery to develop a simple general framework for the \textit{exact} computation of $\Phi$-equilibria in polyhedral games for sets $\Phi$ of linear strategy transformations. This framework parallels that of \citet{Gordon08:No} on no-regret dynamics, but for exact equilibrium computation. Applying this to extensive-form games, we construct the first polynomial-time algorithm for computing exact linear-deviation correlated equilibria in extensive-form games -- a question that had been left open by \citet{Farina2023:Polynomial}.

We believe that having a simple framework to use as a mental model to guide algorithm design is of paramount importance for the advancement of the field. The $\Phi$-regret minimization framework of \citet{Gordon08:No} is indicative of this fact, because it has been key to many interesting results over the years \citep{Morrill2021:Efficient, Farina22:Simple, Anagnostides21:NearOptimal, Farina2023:Polynomial}. Compared to no-regret learning, the problem of exact equilibrium computation has been much less studied (basically only in \citet{Papadimitriou2008:Computing, Jiang2015:Polynomial, Huang2008:Computing}) and we hope that offering a simplified framework will give new insights to advance this front, perhaps aiding in the discovery of new, more practical, algorithms. 

Several key questions remain underinvestigated.
\begin{itemize}
    \item Despite its great theoretical importance, our framework (based on the ellipsoid algorithm) has a polynomial time complexity of rather large degree. Could one devise a more practical alternative while retaining a similar level of generality?

    \item Can our framework be easily extended to general convex strategy spaces (as opposed to polytopes)?

    \item Is there a similar algorithmic framework to compute exact $\Phi$-equilibria in extensive-form games for non-linear transformations $\Phi$? As a concrete direction, can such equilibria be computed efficiently when $\Phi$ is the set of all polynomial deviations on sequence-form strategies? In a recent work, \citet{Zhang2024:Efficient} give parameterized algorithms for minimizing $\Phi$-regret when $\Phi$ is the set of all degree-$k$ polynomial swap deviations or the set of $k$-mediator deviations. It would be interesting to understand whether similar guarantees can be achieved for high-precision computation of these equilibria.

    \item Can ideas similar to those presented in this paper be applied to Markov games?
\end{itemize}

\bibliographystyle{ACM-Reference-Format}

\begin{thebibliography}{44}


\ifx \showCODEN    \undefined \def \showCODEN     #1{\unskip}     \fi
\ifx \showDOI      \undefined \def \showDOI       #1{#1}\fi
\ifx \showISBNx    \undefined \def \showISBNx     #1{\unskip}     \fi
\ifx \showISBNxiii \undefined \def \showISBNxiii  #1{\unskip}     \fi
\ifx \showISSN     \undefined \def \showISSN      #1{\unskip}     \fi
\ifx \showLCCN     \undefined \def \showLCCN      #1{\unskip}     \fi
\ifx \shownote     \undefined \def \shownote      #1{#1}          \fi
\ifx \showarticletitle \undefined \def \showarticletitle #1{#1}   \fi
\ifx \showURL      \undefined \def \showURL       {\relax}        \fi
\providecommand\bibfield[2]{#2}
\providecommand\bibinfo[2]{#2}
\providecommand\natexlab[1]{#1}
\providecommand\showeprint[2][]{arXiv:#2}

\bibitem[Ahmadinejad et~al\mbox{.}(2019)]%
        {Ahmadinejad2019:From}
\bibfield{author}{\bibinfo{person}{AmirMahdi Ahmadinejad},
  \bibinfo{person}{Sina Dehghani}, \bibinfo{person}{MohammadTaghi Hajiaghayi},
  \bibinfo{person}{Brendan Lucier}, \bibinfo{person}{Hamid Mahini}, {and}
  \bibinfo{person}{Saeed Seddighin}.} \bibinfo{year}{2019}\natexlab{}.
\newblock \showarticletitle{From Duels to Battlefields: Computing Equilibria of
  Blotto and Other Games}.
\newblock \bibinfo{journal}{\emph{Mathematics of Operations Research}}
  \bibinfo{volume}{44}, \bibinfo{number}{4} (\bibinfo{year}{2019}),
  \bibinfo{pages}{1304--1325}.
\newblock


\bibitem[Anagnostides et~al\mbox{.}(2022)]%
        {Anagnostides21:NearOptimal}
\bibfield{author}{\bibinfo{person}{Ioannis Anagnostides},
  \bibinfo{person}{Constantinos Daskalakis}, \bibinfo{person}{Gabriele Farina},
  \bibinfo{person}{Maxwell Fishelson}, \bibinfo{person}{Noah Golowich}, {and}
  \bibinfo{person}{Tuomas Sandholm}.} \bibinfo{year}{2022}\natexlab{}.
\newblock \showarticletitle{Near-Optimal No-Regret Learning for Correlated
  Equilibria in Multi-Player General-Sum Games}. In
  \bibinfo{booktitle}{\emph{ACM Symposium on Theory of Computing}}.
\newblock


\bibitem[Assos et~al\mbox{.}(2023)]%
        {Assos2023:Online}
\bibfield{author}{\bibinfo{person}{Angelos Assos}, \bibinfo{person}{Idan
  Attias}, \bibinfo{person}{Yuval Dagan}, \bibinfo{person}{Constantinos
  Daskalakis}, {and} \bibinfo{person}{Maxwell~K. Fishelson}.}
  \bibinfo{year}{2023}\natexlab{}.
\newblock \showarticletitle{Online Learning and Solving Infinite Games with an
  {ERM} Oracle}. In \bibinfo{booktitle}{\emph{The Thirty Sixth Annual
  Conference on Learning Theory, {COLT} 2023, 12-15 July 2023, Bangalore,
  India}} \emph{(\bibinfo{series}{Proceedings of Machine Learning Research},
  Vol.~\bibinfo{volume}{195})}, \bibfield{editor}{\bibinfo{person}{Gergely Neu}
  {and} \bibinfo{person}{Lorenzo Rosasco}} (Eds.). \bibinfo{publisher}{{PMLR}},
  \bibinfo{pages}{274--324}.
\newblock


\bibitem[Aumann(1974)]%
        {Aumann1974:Subjectivity}
\bibfield{author}{\bibinfo{person}{Robert~J Aumann}.}
  \bibinfo{year}{1974}\natexlab{}.
\newblock \showarticletitle{Subjectivity and correlation in randomized
  strategies}.
\newblock \bibinfo{journal}{\emph{Journal of mathematical Economics}}
  \bibinfo{volume}{1}, \bibinfo{number}{1} (\bibinfo{year}{1974}),
  \bibinfo{pages}{67--96}.
\newblock


\bibitem[Blum and Mansour(2007)]%
        {Blum2007:From}
\bibfield{author}{\bibinfo{person}{Avrim Blum} {and} \bibinfo{person}{Yishay
  Mansour}.} \bibinfo{year}{2007}\natexlab{}.
\newblock \showarticletitle{From External to Internal Regret}.
\newblock \bibinfo{journal}{\emph{J. Mach. Learn. Res.}}  \bibinfo{volume}{8}
  (\bibinfo{year}{2007}).
\newblock


\bibitem[Brown(1951)]%
        {Brown1951:Iterative}
\bibfield{author}{\bibinfo{person}{G.W. Brown}.}
  \bibinfo{year}{1951}\natexlab{}.
\newblock \showarticletitle{Iterative Solutions of Games by Fictitious Play}.
\newblock In \bibinfo{booktitle}{\emph{Activity Analysis of Production and
  Allocation}}, \bibfield{editor}{\bibinfo{person}{T.~C. Koopmans}} (Ed.).
  \bibinfo{publisher}{Wiley}, \bibinfo{address}{New York}.
\newblock


\bibitem[Chakrabarti et~al\mbox{.}(2024)]%
        {Chakrabarti2023:Efficient}
\bibfield{author}{\bibinfo{person}{Darshan Chakrabarti},
  \bibinfo{person}{Gabriele Farina}, {and} \bibinfo{person}{Christian Kroer}.}
  \bibinfo{year}{2024}\natexlab{}.
\newblock \showarticletitle{Efficient Learning in Polyhedral Games via Best
  Response Oracles}. In \bibinfo{booktitle}{\emph{AAAI Conference on Artificial
  Intelligence (AAAI)}}.
\newblock


\bibitem[Dagan et~al\mbox{.}(2023)]%
        {Dagan2023:From}
\bibfield{author}{\bibinfo{person}{Yuval Dagan}, \bibinfo{person}{Constantinos
  Daskalakis}, \bibinfo{person}{Maxwell Fishelson}, {and} \bibinfo{person}{Noah
  Golowich}.} \bibinfo{year}{2023}\natexlab{}.
\newblock \bibinfo{title}{From External to Swap Regret 2.0: An Efficient
  Reduction and Oblivious Adversary for Large Action Spaces}.
\newblock
\newblock
\showeprint[arxiv]{2310.19786}~[cs.LG]


\bibitem[Daskalakis et~al\mbox{.}(2009)]%
        {Daskalakis2009:The}
\bibfield{author}{\bibinfo{person}{Constantinos Daskalakis},
  \bibinfo{person}{Paul~W. Goldberg}, {and} \bibinfo{person}{Christos~H.
  Papadimitriou}.} \bibinfo{year}{2009}\natexlab{}.
\newblock \showarticletitle{The complexity of computing a Nash equilibrium}.
\newblock \bibinfo{journal}{\emph{Commun. ACM}} \bibinfo{volume}{52},
  \bibinfo{number}{2} (\bibinfo{date}{feb} \bibinfo{year}{2009}),
  \bibinfo{pages}{89–97}.
\newblock
\showISSN{0001-0782}


\bibitem[FAIR et~al\mbox{.}(2022)]%
        {Meta2022:Human}
\bibfield{author}{\bibinfo{person}{FAIR}, \bibinfo{person}{Anton Bakhtin},
  \bibinfo{person}{Noam Brown}, \bibinfo{person}{Emily Dinan},
  \bibinfo{person}{Gabriele Farina}, \bibinfo{person}{Colin Flaherty},
  \bibinfo{person}{Daniel Fried}, \bibinfo{person}{Andrew Goff},
  \bibinfo{person}{Jonathan Gray}, \bibinfo{person}{Hengyuan Hu},
  \bibinfo{person}{Athul~Paul Jacob}, \bibinfo{person}{Mojtaba Komeili},
  \bibinfo{person}{Karthik Konath}, \bibinfo{person}{Minae Kwon},
  \bibinfo{person}{Adam Lerer}, \bibinfo{person}{Mike Lewis},
  \bibinfo{person}{Alexander~H. Miller}, \bibinfo{person}{Sasha Mitts},
  \bibinfo{person}{Adithya Renduchintala}, \bibinfo{person}{Stephen Roller},
  \bibinfo{person}{Dirk Rowe}, \bibinfo{person}{Weiyan Shi},
  \bibinfo{person}{Joe Spisak}, \bibinfo{person}{Alexander Wei},
  \bibinfo{person}{David Wu}, \bibinfo{person}{Hugh Zhang}, {and}
  \bibinfo{person}{Markus Zijlstra}.} \bibinfo{year}{2022}\natexlab{}.
\newblock \showarticletitle{Human-level play in the game of Diplomacy by
  combining language models with strategic reasoning}.
\newblock \bibinfo{journal}{\emph{Science}} \bibinfo{volume}{378},
  \bibinfo{number}{6624} (\bibinfo{year}{2022}), \bibinfo{pages}{1067--1074}.
\newblock


\bibitem[Farina et~al\mbox{.}(2020)]%
        {Farina20:Coarse}
\bibfield{author}{\bibinfo{person}{Gabriele Farina}, \bibinfo{person}{Tommaso
  Bianchi}, {and} \bibinfo{person}{Tuomas Sandholm}.}
  \bibinfo{year}{2020}\natexlab{}.
\newblock \showarticletitle{Coarse Correlation in Extensive-Form Games}. In
  \bibinfo{booktitle}{\emph{AAAI Conference on Artificial Intelligence}}.
\newblock


\bibitem[Farina et~al\mbox{.}(2022a)]%
        {Farina22:Simple}
\bibfield{author}{\bibinfo{person}{Gabriele Farina}, \bibinfo{person}{Andrea
  Celli}, \bibinfo{person}{Alberto Marchesi}, {and} \bibinfo{person}{Nicola
  Gatti}.} \bibinfo{year}{2022}\natexlab{a}.
\newblock \showarticletitle{Simple Uncoupled No-regret Learning Dynamics for
  Extensive-form Correlated Equilibrium}.
\newblock \bibinfo{journal}{\emph{J. ACM}} \bibinfo{volume}{69},
  \bibinfo{number}{6} (\bibinfo{year}{2022}).
\newblock


\bibitem[Farina et~al\mbox{.}(2022b)]%
        {Farina22:Kernelized}
\bibfield{author}{\bibinfo{person}{Gabriele Farina}, \bibinfo{person}{Chung-Wei
  Lee}, \bibinfo{person}{Haipeng Luo}, {and} \bibinfo{person}{Christian
  Kroer}.} \bibinfo{year}{2022}\natexlab{b}.
\newblock \showarticletitle{Kernelized Multiplicative Weights for
  0/1-Polyhedral Games: Bridging the Gap Between Learning in Extensive-Form and
  Normal-Form Games}. In \bibinfo{booktitle}{\emph{International Conference on
  Machine Learning}}.
\newblock


\bibitem[Farina and Pipis(2023)]%
        {Farina2023:Polynomial}
\bibfield{author}{\bibinfo{person}{Gabriele Farina} {and}
  \bibinfo{person}{Charilaos Pipis}.} \bibinfo{year}{2023}\natexlab{}.
\newblock \showarticletitle{Polynomial-Time {L}inear-{S}wap {R}egret
  {M}inimization in {I}mperfect-{I}nformation {S}equential {G}ames}. In
  \bibinfo{booktitle}{\emph{Thirty-seventh Conference on Neural Information
  Processing Systems}}.
\newblock


\bibitem[Gemp et~al\mbox{.}(2024)]%
        {Gemp2024:States}
\bibfield{author}{\bibinfo{person}{Ian Gemp}, \bibinfo{person}{Yoram Bachrach},
  \bibinfo{person}{Marc Lanctot}, \bibinfo{person}{Roma Patel},
  \bibinfo{person}{Vibhavari Dasagi}, \bibinfo{person}{Luke Marris},
  \bibinfo{person}{Georgios Piliouras}, \bibinfo{person}{Siqi Liu}, {and}
  \bibinfo{person}{Karl Tuyls}.} \bibinfo{year}{2024}\natexlab{}.
\newblock \bibinfo{title}{States as Strings as Strategies: Steering Language
  Models with Game-Theoretic Solvers}.
\newblock
\newblock
\showeprint[arxiv]{2402.01704}~[cs.CL]


\bibitem[Gidel et~al\mbox{.}(2017)]%
        {Gidel2017:Frank}
\bibfield{author}{\bibinfo{person}{Gauthier Gidel}, \bibinfo{person}{Tony
  Jebara}, {and} \bibinfo{person}{Simon Lacoste-Julien}.}
  \bibinfo{year}{2017}\natexlab{}.
\newblock \showarticletitle{{Frank-Wolfe Algorithms for Saddle Point
  Problems}}. In \bibinfo{booktitle}{\emph{Proceedings of the 20th
  International Conference on Artificial Intelligence and Statistics}}
  \emph{(\bibinfo{series}{Proceedings of Machine Learning Research},
  Vol.~\bibinfo{volume}{54})}, \bibfield{editor}{\bibinfo{person}{Aarti Singh}
  {and} \bibinfo{person}{Jerry Zhu}} (Eds.). \bibinfo{publisher}{PMLR},
  \bibinfo{pages}{362--371}.
\newblock


\bibitem[Goldberg and Marmolejo-Cossío(2021)]%
        {Goldberg2021:Learning}
\bibfield{author}{\bibinfo{person}{Paul~W. Goldberg} {and}
  \bibinfo{person}{Francisco~J. Marmolejo-Cossío}.}
  \bibinfo{year}{2021}\natexlab{}.
\newblock \showarticletitle{Learning {Convex} {Partitions} and {Computing}
  {Game}-theoretic {Equilibria} from {Best}-response {Queries}}.
\newblock \bibinfo{journal}{\emph{ACM Transactions on Economics and
  Computation}} \bibinfo{volume}{9}, \bibinfo{number}{1}
  (\bibinfo{year}{2021}), \bibinfo{pages}{3:1--3:36}.
\newblock


\bibitem[Gordon et~al\mbox{.}(2008)]%
        {Gordon08:No}
\bibfield{author}{\bibinfo{person}{Geoffrey~J Gordon}, \bibinfo{person}{Amy
  Greenwald}, {and} \bibinfo{person}{Casey Marks}.}
  \bibinfo{year}{2008}\natexlab{}.
\newblock \showarticletitle{No-regret learning in convex games}. In
  \bibinfo{booktitle}{\emph{International Conference on Machine learning}}.
  \bibinfo{pages}{360--367}.
\newblock


\bibitem[Gr{\"o}tschel et~al\mbox{.}(1993)]%
        {Grotschel1993:Geometric}
\bibfield{author}{\bibinfo{person}{Martin Gr{\"o}tschel},
  \bibinfo{person}{L{\'a}szl{\'o} Lov{\'a}sz}, {and} \bibinfo{person}{Alexander
  Schrijver}.} \bibinfo{year}{1993}\natexlab{}.
\newblock \bibinfo{booktitle}{\emph{Geometric Algorithms and Combinatorial
  Optimization}}.
\newblock \bibinfo{publisher}{Springer Berlin, Heidelberg}.
\newblock


\bibitem[Hart and Mas-Colell(2000)]%
        {Hart00:Simple}
\bibfield{author}{\bibinfo{person}{Sergiu Hart} {and} \bibinfo{person}{Andreu
  Mas-Colell}.} \bibinfo{year}{2000}\natexlab{}.
\newblock \showarticletitle{A Simple Adaptive Procedure Leading to Correlated
  Equilibrium}.
\newblock \bibinfo{journal}{\emph{Econometrica}} \bibinfo{volume}{68},
  \bibinfo{number}{5} (\bibinfo{year}{2000}), \bibinfo{pages}{1127--1150}.
\newblock


\bibitem[Hart and Schmeidler(1989)]%
        {Hart1989:Existence}
\bibfield{author}{\bibinfo{person}{Sergiu Hart} {and} \bibinfo{person}{David
  Schmeidler}.} \bibinfo{year}{1989}\natexlab{}.
\newblock \showarticletitle{Existence of Correlated Equilibria}.
\newblock \bibinfo{journal}{\emph{Mathematics of Operations Research}}
  \bibinfo{volume}{14}, \bibinfo{number}{1} (\bibinfo{year}{1989}),
  \bibinfo{pages}{18--25}.
\newblock


\bibitem[Hazan and Kale(2007)]%
        {Hazan2007:Computational}
\bibfield{author}{\bibinfo{person}{Elad Hazan} {and} \bibinfo{person}{Satyen
  Kale}.} \bibinfo{year}{2007}\natexlab{}.
\newblock \showarticletitle{Computational Equivalence of Fixed Points and No
  Regret Algorithms, and Convergence to Equilibria}. In
  \bibinfo{booktitle}{\emph{Advances in Neural Information Processing
  Systems}}, \bibfield{editor}{\bibinfo{person}{J.~Platt},
  \bibinfo{person}{D.~Koller}, \bibinfo{person}{Y.~Singer}, {and}
  \bibinfo{person}{S.~Roweis}} (Eds.), Vol.~\bibinfo{volume}{20}.
  \bibinfo{publisher}{Curran Associates, Inc.}
\newblock


\bibitem[Huang and von Stengel(2008)]%
        {Huang2008:Computing}
\bibfield{author}{\bibinfo{person}{Wan Huang} {and} \bibinfo{person}{Bernhard
  von Stengel}.} \bibinfo{year}{2008}\natexlab{}.
\newblock \showarticletitle{Computing an extensive-form correlated equilibrium
  in polynomial time}. In \bibinfo{booktitle}{\emph{International Workshop on
  Internet and Network Economics}}. Springer, \bibinfo{pages}{506--513}.
\newblock


\bibitem[Jiang and Leyton-Brown(2015)]%
        {Jiang2015:Polynomial}
\bibfield{author}{\bibinfo{person}{Albert~Xin Jiang} {and}
  \bibinfo{person}{Kevin Leyton-Brown}.} \bibinfo{year}{2015}\natexlab{}.
\newblock \showarticletitle{Polynomial-time computation of exact correlated
  equilibrium in compact games}.
\newblock \bibinfo{journal}{\emph{Games and Economic Behavior}}
  \bibinfo{volume}{91} (\bibinfo{year}{2015}), \bibinfo{pages}{347--359}.
\newblock


\bibitem[Kiekintveld et~al\mbox{.}(2009)]%
        {Kiekintveld2009:Computing}
\bibfield{author}{\bibinfo{person}{Christopher Kiekintveld},
  \bibinfo{person}{Manish Jain}, \bibinfo{person}{Jason Tsai},
  \bibinfo{person}{James Pita}, \bibinfo{person}{Fernando Ord\'{o}\~{n}ez},
  {and} \bibinfo{person}{Milind Tambe}.} \bibinfo{year}{2009}\natexlab{}.
\newblock \showarticletitle{Computing optimal randomized resource allocations
  for massive security games}. In \bibinfo{booktitle}{\emph{Proceedings of The
  8th International Conference on Autonomous Agents and Multiagent Systems -
  Volume 1}} (Budapest, Hungary) \emph{(\bibinfo{series}{AAMAS '09})}.
  \bibinfo{publisher}{International Foundation for Autonomous Agents and
  Multiagent Systems}, \bibinfo{address}{Richland, SC},
  \bibinfo{pages}{689–696}.
\newblock
\showISBNx{9780981738161}


\bibitem[Koller et~al\mbox{.}(1996)]%
        {Koller96:Efficient}
\bibfield{author}{\bibinfo{person}{Daphne Koller}, \bibinfo{person}{Nimrod
  Megiddo}, {and} \bibinfo{person}{Bernhard {von Stengel}}.}
  \bibinfo{year}{1996}\natexlab{}.
\newblock \showarticletitle{Efficient Computation of Equilibria for Extensive
  Two-Person Games}.
\newblock \bibinfo{journal}{\emph{Games and Economic Behavior}}
  \bibinfo{volume}{14}, \bibinfo{number}{2} (\bibinfo{year}{1996}),
  \bibinfo{pages}{247--259}.
\newblock


\bibitem[Koolen et~al\mbox{.}(2010)]%
        {Koolen10:Hedging}
\bibfield{author}{\bibinfo{person}{Wouter~M Koolen}, \bibinfo{person}{Manfred~K
  Warmuth}, \bibinfo{person}{Jyrki Kivinen}, {et~al\mbox{.}}}
  \bibinfo{year}{2010}\natexlab{}.
\newblock \showarticletitle{Hedging Structured Concepts}. In
  \bibinfo{booktitle}{\emph{COLT}}. Citeseer, \bibinfo{pages}{93--105}.
\newblock


\bibitem[Lanctot et~al\mbox{.}(2017)]%
        {Lanctot2017:Unified}
\bibfield{author}{\bibinfo{person}{Marc Lanctot},
  \bibinfo{person}{Vin{\'{\i}}cius~Flores Zambaldi}, \bibinfo{person}{Audrunas
  Gruslys}, \bibinfo{person}{Angeliki Lazaridou}, \bibinfo{person}{Karl Tuyls},
  \bibinfo{person}{Julien P{\'{e}}rolat}, \bibinfo{person}{David Silver}, {and}
  \bibinfo{person}{Thore Graepel}.} \bibinfo{year}{2017}\natexlab{}.
\newblock \showarticletitle{A Unified Game-Theoretic Approach to Multiagent
  Reinforcement Learning}. In \bibinfo{booktitle}{\emph{Advances in Neural
  Information Processing Systems 30: Annual Conference on Neural Information
  Processing Systems 2017, December 4-9, 2017, Long Beach, CA, {USA}}},
  \bibfield{editor}{\bibinfo{person}{Isabelle Guyon}, \bibinfo{person}{Ulrike
  von Luxburg}, \bibinfo{person}{Samy Bengio}, \bibinfo{person}{Hanna~M.
  Wallach}, \bibinfo{person}{Rob Fergus}, \bibinfo{person}{S.~V.~N.
  Vishwanathan}, {and} \bibinfo{person}{Roman Garnett}} (Eds.).
  \bibinfo{pages}{4190--4203}.
\newblock


\bibitem[Marks(2008)]%
        {Marks2008:No}
\bibfield{author}{\bibinfo{person}{C. Marks}.} \bibinfo{year}{2008}\natexlab{}.
\newblock \emph{\bibinfo{title}{No-regret learning and game-theoretic
  equilibria}}.
\newblock \bibinfo{thesistype}{Ph.\,D. Dissertation}. \bibinfo{school}{Brown
  University}, \bibinfo{address}{Providence, RI}.
\newblock


\bibitem[Morrill et~al\mbox{.}(2021)]%
        {Morrill2021:Efficient}
\bibfield{author}{\bibinfo{person}{Dustin Morrill}, \bibinfo{person}{Ryan
  D'Orazio}, \bibinfo{person}{Marc Lanctot}, \bibinfo{person}{James~R. Wright},
  \bibinfo{person}{Michael Bowling}, {and} \bibinfo{person}{Amy~R. Greenwald}.}
  \bibinfo{year}{2021}\natexlab{}.
\newblock \showarticletitle{Efficient Deviation Types and Learning for
  Hindsight Rationality in Extensive-Form Games}. In
  \bibinfo{booktitle}{\emph{Proceedings of the 38th International Conference on
  Machine Learning, ICML 2021, 18-24 July 2021, Virtual Event}}
  \emph{(\bibinfo{series}{Proceedings of Machine Learning Research},
  Vol.~\bibinfo{volume}{139})}. \bibinfo{publisher}{PMLR},
  \bibinfo{pages}{7818--7828}.
\newblock


\bibitem[Moulin and Vial(1978)]%
        {Moulin1978:Strategically}
\bibfield{author}{\bibinfo{person}{H. Moulin} {and} \bibinfo{person}{J.~P.
  Vial}.} \bibinfo{year}{1978}\natexlab{}.
\newblock \showarticletitle{Strategically zero-sum games: The class of games
  whose completely mixed equilibria cannot be improved upon}.
\newblock \bibinfo{journal}{\emph{Int. J. Game Theory}} \bibinfo{volume}{7},
  \bibinfo{number}{3–4} (\bibinfo{date}{sep} \bibinfo{year}{1978}),
  \bibinfo{pages}{201–221}.
\newblock


\bibitem[Neumann(1928)]%
        {vNeumann1928}
\bibfield{author}{\bibinfo{person}{J.~von Neumann}.}
  \bibinfo{year}{1928}\natexlab{}.
\newblock \showarticletitle{Zur Theorie der Gesellschaftsspiele}.
\newblock \bibinfo{journal}{\emph{Math. Ann.}} \bibinfo{volume}{100},
  \bibinfo{number}{1} (\bibinfo{year}{1928}), \bibinfo{pages}{295--320}.
\newblock


\bibitem[Papadimitriou and Roughgarden(2008)]%
        {Papadimitriou2008:Computing}
\bibfield{author}{\bibinfo{person}{Christos~H. Papadimitriou} {and}
  \bibinfo{person}{Tim Roughgarden}.} \bibinfo{year}{2008}\natexlab{}.
\newblock \showarticletitle{Computing {C}orrelated {E}quilibria in
  {M}ulti-{P}layer {G}ames}.
\newblock \bibinfo{journal}{\emph{J. ACM}} \bibinfo{volume}{55},
  \bibinfo{number}{3} (\bibinfo{year}{2008}).
\newblock


\bibitem[Peng and Rubinstein(2023)]%
        {Peng2023:Fast}
\bibfield{author}{\bibinfo{person}{Binghui Peng} {and} \bibinfo{person}{Aviad
  Rubinstein}.} \bibinfo{year}{2023}\natexlab{}.
\newblock \bibinfo{title}{Fast swap regret minimization and applications to
  approximate correlated equilibria}.
\newblock
\newblock
\showeprint[arxiv]{2310.19647}~[cs.GT]


\bibitem[Romanovskii(1962)]%
        {Romanovskii62:Reduction}
\bibfield{author}{\bibinfo{person}{I. Romanovskii}.}
  \bibinfo{year}{1962}\natexlab{}.
\newblock \showarticletitle{{Reduction of a Game with Complete Memory to a
  Matrix Game}}.
\newblock \bibinfo{journal}{\emph{Soviet Mathematics}}  \bibinfo{volume}{3}
  (\bibinfo{year}{1962}).
\newblock


\bibitem[Rubinstein(2015)]%
        {Rubinstein2015:Inapproximability}
\bibfield{author}{\bibinfo{person}{Aviad Rubinstein}.}
  \bibinfo{year}{2015}\natexlab{}.
\newblock \showarticletitle{Inapproximability of Nash Equilibrium}. In
  \bibinfo{booktitle}{\emph{Proceedings of the Forty-Seventh Annual ACM
  Symposium on Theory of Computing}} (Portland, Oregon, USA)
  \emph{(\bibinfo{series}{STOC '15})}. \bibinfo{publisher}{Association for
  Computing Machinery}, \bibinfo{address}{New York, NY, USA},
  \bibinfo{pages}{409–418}.
\newblock


\bibitem[Rubinstein(2016)]%
        {Rubinstein2016:Settling}
\bibfield{author}{\bibinfo{person}{Aviad Rubinstein}.}
  \bibinfo{year}{2016}\natexlab{}.
\newblock \showarticletitle{Settling the Complexity of Computing Approximate
  Two-Player Nash Equilibria}. In \bibinfo{booktitle}{\emph{2016 IEEE 57th
  Annual Symposium on Foundations of Computer Science (FOCS)}}.
  \bibinfo{pages}{258--265}.
\newblock


\bibitem[Takimoto and Warmuth(2003)]%
        {Takimoto03:Path}
\bibfield{author}{\bibinfo{person}{Eiji Takimoto} {and}
  \bibinfo{person}{Manfred~K Warmuth}.} \bibinfo{year}{2003}\natexlab{}.
\newblock \showarticletitle{Path kernels and multiplicative updates}.
\newblock \bibinfo{journal}{\emph{The Journal of Machine Learning Research}}
  \bibinfo{volume}{4} (\bibinfo{year}{2003}), \bibinfo{pages}{773--818}.
\newblock


\bibitem[{von Stengel}(1996)]%
        {vonStengel96:Efficient}
\bibfield{author}{\bibinfo{person}{Bernhard {von Stengel}}.}
  \bibinfo{year}{1996}\natexlab{}.
\newblock \showarticletitle{Efficient Computation of Behavior Strategies}.
\newblock \bibinfo{journal}{\emph{Games and Economic Behavior}}
  \bibinfo{volume}{14}, \bibinfo{number}{2} (\bibinfo{year}{1996}),
  \bibinfo{pages}{220--246}.
\newblock


\bibitem[von Stengel and Forges(2008)]%
        {vonStengel2008}
\bibfield{author}{\bibinfo{person}{B. von Stengel} {and} \bibinfo{person}{F.
  Forges}.} \bibinfo{year}{2008}\natexlab{}.
\newblock \showarticletitle{Extensive-form correlated equilibrium: Definition
  and computational complexity}.
\newblock \bibinfo{journal}{\emph{Mathematics of Operations Research}}
  \bibinfo{volume}{33}, \bibinfo{number}{4} (\bibinfo{year}{2008}),
  \bibinfo{pages}{1002--1022}.
\newblock


\bibitem[Xu(2016)]%
        {Xu2016:The}
\bibfield{author}{\bibinfo{person}{Haifeng Xu}.}
  \bibinfo{year}{2016}\natexlab{}.
\newblock \showarticletitle{The Mysteries of Security Games: Equilibrium
  Computation Becomes Combinatorial Algorithm Design}. In
  \bibinfo{booktitle}{\emph{Proceedings of the 2016 ACM Conference on Economics
  and Computation}} (Maastricht, The Netherlands) \emph{(\bibinfo{series}{EC
  '16})}. \bibinfo{publisher}{Association for Computing Machinery},
  \bibinfo{address}{New York, NY, USA}, \bibinfo{pages}{497–514}.
\newblock
\showISBNx{9781450339360}


\bibitem[Xu et~al\mbox{.}(2014)]%
        {Xu2014:Solving}
\bibfield{author}{\bibinfo{person}{Haifeng Xu}, \bibinfo{person}{Fei Fang},
  \bibinfo{person}{Albert Jiang}, \bibinfo{person}{Vincent Conitzer},
  \bibinfo{person}{Shaddin Dughmi}, {and} \bibinfo{person}{Milind Tambe}.}
  \bibinfo{year}{2014}\natexlab{}.
\newblock \showarticletitle{Solving Zero-Sum Security Games in Discretized
  Spatio-Temporal Domains}.
\newblock \bibinfo{journal}{\emph{Proceedings of the AAAI Conference on
  Artificial Intelligence}} \bibinfo{volume}{28}, \bibinfo{number}{1}
  (\bibinfo{date}{Jun.} \bibinfo{year}{2014}).
\newblock


\bibitem[Zhang et~al\mbox{.}(2024a)]%
        {Zhang2024:Efficient}
\bibfield{author}{\bibinfo{person}{Brian~Hu Zhang}, \bibinfo{person}{Ioannis
  Anagnostides}, \bibinfo{person}{Gabriele Farina}, {and}
  \bibinfo{person}{Tuomas Sandholm}.} \bibinfo{year}{2024}\natexlab{a}.
\newblock \bibinfo{title}{Efficient $\Phi$-Regret Minimization with Low-Degree
  Swap Deviations in Extensive-Form Games}.
\newblock
\newblock
\showeprint[arxiv]{2402.09670}~[cs.GT]


\bibitem[Zhang et~al\mbox{.}(2024b)]%
        {Zhang2024:Mediator}
\bibfield{author}{\bibinfo{person}{Brian~H. Zhang}, \bibinfo{person}{Gabriele
  Farina}, {and} \bibinfo{person}{Tuomas Sandholm}.}
  \bibinfo{year}{2024}\natexlab{b}.
\newblock \showarticletitle{Mediator {I}nterpretation and {F}aster {L}earning
  {A}lgorithms for {L}inear {C}orrelated {E}quilibria in {G}eneral {S}equential
  {G}ames}. In \bibinfo{booktitle}{\emph{The Twelfth International Conference
  on Learning Representations}}.
\newblock


\end{thebibliography}

\appendix

\section{Omitted proofs}
\label{app:proofs}

\lempolylinprog*
\begin{proof}
    First we show that $\mathcal{P} = \mathcal{P}'$, where
    \[
        \mathcal{P}' := \left\{ \vec{u} \in \mathcal{U} \bigm\vert (\hat{\vec q}^\top \mat{A}) \vec{u} \leq c \ \forall {\hat{\vec q} \in V(\mathcal{Q})}
        \right\},
    \]
    is the polytope defined by finitely many inequality constraints, corresponding to the vertices of $\mathcal{Q}$. Combining this with the assumption that each such inequality has encoding length at most $\varphi$, the result follows immediately. It remains to prove the desired set equality:

    \begin{itemize}
        \item Case $\mathcal{P} \subseteq \mathcal{P}'$:
        \[
            \vec{u} \in \mathcal{P} \implies \vec{u} \in \mathcal{U}, \  \max_{\vec q \in \mathcal{Q}} \vec{q}^\top \mat{A} \vec{u} \leq c \implies \vec{u} \in \mathcal{U}, \  (\hat{\vec{q}}^\top \mat{A}) \vec{u} \leq c \ \forall {\hat{\vec{q}} \in V(\mathcal{Q})},
        \]
        where the last implication follows from the fact that for all $\hat{\vec{q}} \in V(\mathcal{Q}) \subseteq \mathcal{Q}$,
        \[
            (\hat{\vec{q}}^\top \mat{A}) \vec{u} \leq \max_{\vec q \in \mathcal{Q}} \vec{q}^\top \mat{A} \vec{u} \leq c.
        \]

        \item Case $\mathcal{P} \supseteq \mathcal{P}'$:
        By definition, any point $\vec q \in \mathcal{Q}$ can be written as the convex combination of all vertices $\vec q = \sum_i^K \lambda_i \hat{\vec q}_i$. Thus, we have
        \[
            \vec u \in \mathcal{P}' \implies \vec{u} \in \mathcal{U}, \  (\hat{\vec{q}}^\top \mat{A}) \vec{u} \leq c \ \forall {\hat{\vec{q}} \in V(\mathcal{Q})} \implies \vec{u} \in \mathcal{U}, \ \max_{\vec q \in \mathcal{Q}} \vec{q}^\top \mat{A} \vec{u} \leq c,
        \]
        where the last implication holds because for any $\vec q \in \mathcal{Q}$,
        \[
            \vec{q}^\top \mat{A} \vec{u} = \sum_i^K \lambda_i \hat{\vec q}_i^\top \mat{A} \vec{u} \leq \sum_i^K \lambda_i c.
        \]
    \end{itemize}
    This concludes the proof.
\end{proof}

\lemfarkas*
\begin{proof}
    I) We first show that (1) and (2) cannot be true simultaneously. Assume otherwise and let $\hat{\vec{x}} \in \mathcal{X}$, $\hat{\vec{y}} \in \bbR_+ \mathcal{Y}$ be values that satisfy (1) and (2) respectively. Since $\hat{\vec{y}}$ belongs to the conic hull of $\mathcal{Y}$ it must be $\hat{\vec{y}} = k \vec{y}'$ for some $k > 0$ and $\vec{y}' \in \mathcal{Y}$. Thus,
    \[
        \max_{\vec{x} \in \mathcal{X}} \vec{x}^\top \mat{A} \hat{\vec{y}} \leq -1 \implies \max_{\vec{x} \in \mathcal{X}} \vec{x}^\top \mat{A} \vec{y}' \leq -\frac{1}{k}.
    \]
    Additionally, it holds
    \[
        0 \leq \min_{\vec{y} \in \mathcal{Y}} \hat{\vec{x}}^\top \mat{A} \vec{y} \leq \hat{\vec{x}}^\top \mat{A} \vec{y}' \leq \max_{\vec{x} \in \mathcal{X}} \vec{x}^\top \mat{A} \vec{y}' \leq -\frac{1}{k},
    \]
    which is a contradiction. Thus, the statements (1) and (2) cannot be true simultaneously.

    II) We now proceed to prove that when (2) is false then (1) must be true. We begin by showing that (2) being false implies that for any $\gamma > 0$ there does not exist any $\vec{y} \in \bbR_+ \mathcal{Y}$ such that $\displaystyle \max_{\vec{x} \in \mathcal{X}} \vec{x}^\top \mat{A} \vec{y} \leq -\gamma$. Suppose otherwise; then $\vec{y}' = \vec{y} / \gamma$ is a multiple of an element of $\bbR_+ \mathcal{Y}$ and thus $\vec{y}' \in \bbR_+ \mathcal{Y}$. Furthermore,
    \[
        \max_{\vec{x} \in \mathcal{X}} \vec{x}^\top \mat{A} \vec{y}' = \frac{1}{\gamma} \max_{\vec{x} \in \mathcal{X}} \vec{x}^\top  \mat{A} \vec{y} \leq -1,
    \]
    which is a contradiction because we have assumed that (2) is false. It directly follows that
    \[
        \min_{\vec{y} \in \mathcal{Y}} \max_{\vec{x} \in \mathcal{X}} \vec{x}^\top  \mat{A} \vec{y} \geq 0.
    \]
    By the minimax theorem it also holds
    \[
        \max_{\vec{x} \in \mathcal{X}} \min_{\vec{y} \in \mathcal{Y}} \vec{x}^\top  \mat{A} \vec{y} \geq 0 \implies \exists_{\vec{x} \in \mathcal{X}}:\  \min_{\vec{y} \in \mathcal{Y}} \vec{x}^\top  \mat{A} \vec{y} \geq 0
    \]
    and thus, statement (1) is true.

    III) Finally, we need to prove the inverse direction; when (1) is false then (2) must be true. This is trivial because if (2) was false, then by II) we would have that (1) is true, which contradicts the assumption.
\end{proof}

\end{document}